\documentclass{article}
\usepackage{amssymb}
\usepackage{amsmath}
\usepackage{amsthm}
\usepackage{amsfonts}
\usepackage{graphicx,epsfig,graphics}
\usepackage[left=2cm,top=2.5cm,right=2.5cm,nohead,foot=1cm]{geometry}
\usepackage{subfigure}
\usepackage{color}
\graphicspath{ {./plots/} }

\usepackage{epstopdf} 

\newcommand{\nB}{|B|}
\newcommand{\nA}{|A|}

\newtheorem{theorem}{Theorem}

\newtheorem{comment}[theorem]{Comment}

\newtheorem{corollary}[theorem]{Corollary}

\newtheorem{example}[theorem]{Example}

\newtheorem{lemma}[theorem]{Lemma}

\newtheorem{proposition}[theorem]{Proposition}
\newtheorem{remark}[theorem]{Remark}

\renewenvironment{proof}[1][Proof]{\noindent\textbf{#1.} }{\
\rule{0.5em}{0.5em}}
\numberwithin{equation}{section}
\numberwithin{theorem}{section}
\input{tcilatex}

\newcommand{\be}{\begin{equation}}
\newcommand{\ee}{\end{equation}}
\newcommand{\Prob}{\mathbb{P}}
\newcommand{\R}{\mathbb{R}}
\newcommand{\esp}{\mathbb{E}}
\newcommand{\K}{\mathcal{K}}
\newcommand{\eps}{\varepsilon}
\newcommand{\Horm}{H\"ormander }
\newcommand{\Hold}{H\"older}
\newcommand{\Gy}{Gy\"ongy}

\title{Varadhan's formula, conditioned diffusions, and local volatilities
\thanks{
Corresponding author: Stefano De Marco, Ecole Polytechnique-CMAP, Route de Saclay, 91128 Palaiseau Cedex, France. demarco@cmap.polytechnique.fr
\newline \emph{Key words and phrases}: conditional density asymptotics, local volatility, stochastic volatility, large deviations.
\newline \emph{2010 Mathematics Subject Classification}: AMS 91G20, 91G80, 60H30, 65C30.
}}
\author{Stefano De Marco, Peter Friz
\\
Ecole Polytechnique, TU and WIAS Berlin}
\renewcommand\footnotemark{}
\date{}

\begin{document}
\maketitle

\begin{abstract}
\noindent
Motivated by marginals-mimicking results for It\^o processes \cite{Dupire,GyonLocVol} via SDEs and by their applications to volatility modeling in finance, we discuss the weak convergence of the law of a hypoelliptic diffusions conditioned to belong to
a target affine subspace at final time, namely $\mathcal{L}(Z_t|Y_t = y)$ if $X_{\cdot}=(Y_\cdot,Z_{\cdot})$.
To do so, we revisit Varadhan-type estimates in a small-noise (as opposed to small-time) regime, studying the density of the lower-dimensional component $Y$.	
The application to stochastic volatility models include the small-time and, for certain models, the large-strike asymptotics of the \Gy--Dupire local volatility function. The final product are asymptotic formulae that can (i) motivate parameterizations of the local volatility surface and (ii) be used to extrapolate local volatilities in a given model.
\end{abstract}

\section{Introduction}

Consider an $n$-dimensional diffusion process given by the solution of
\be \label{e:SDEintro}
dX_t = b\left( X_t\right) dt + \sum_{j=1}^d \sigma_j \left( X_t\right) dW^j_t, \ X_{0}=x_0 \in \R^n.
\ee
Applications to finance suggest a splitting of the state space, say $X=(Y,Z) \in \R^l \times \R^{n-l} \cong \R^n$, where
$Y$ is the main process of interest (for instance: price or log-price of an asset) and $Z$ some auxiliary process (for instance: stochastic volatility, possibly multi-dimensional). There is a massive amount of literature concerning $p_t$, the probability distribution function of $X_t$ at small times $t$. In the {\it elliptic} case, that is when $span \{\sigma_1, ..., \sigma_d\} = \R^n$ such investigations go back to Varadhan (``$2t \log p_t(x,y) \sim d^2(x,y)$'') and then Molchanov \cite{Mol75} for full expansions of $p_t$. The hypoelliptic situation (assuming the strong H\"ormander condition $Lie\{\sigma_1, ..., \sigma_d\} = \R^n$) was then studied by Azencott, Bismut, Leandre, Ben Arous,...  (the function $d$ is then be interpreted as control distance associated to the diffusion vector fields.)

Similar results were recently obtained for $f_t$, the density of $Y_t$ (that is, a marginal density of $X_t$) by Deuschel et al. see \cite{DFJV:I,DFJV:II}, improving on earlier works of Takanobu--Watanabe \cite{TakaWat}. We postpone a detailed comparison of \cite{TakaWat} and our results to Section \ref{sec:TW} below.
%Let us just insist at this stage that the results of this paper cannot be obtained as application of \cite{TakaWat} (or, for that matter, \cite{DFJV:I,DFJV:II}): in all these papers a key role is played by (necessary!) non-degeneracy condition which then leads to an asymptotic expansions of $f_t$. Such conditions can be difficult to check in practice (or may fail to hold) while the results in this paper -- morally of ``large deviation" and ``law of large number" type -- do not require such conditions.

Our main results are: (i) a Varadhan formula for $f_t$ in the short time limit, which is seen to be valid in great generality (without the need to check for the non-degeneracy conditions that appear as assumptions in \cite{TakaWat,DFJV:I,DFJV:II}\footnote{As was seen in \cite{DFJV:II}, checking non-focality requires a non--trivial analysis of certain Hamiltonian systems. We also note that focality actually happens in reasonably simple situations, e.g. in the context of a 2-dimensional Black--Scholes basket, as was seen in \cite{BFLbaskets}.}) and then (ii) a limit theorem for $Z$ conditioned on the value of $Y$. As far as we know, even in the elliptic case our results, concerning the marginals of $X_t$, are new.
The limit here may again be short time or, more generally, small noise. In fact, the small noise situations poses new difficulties (for instance, in a strictly hypoelliptic setting Varadhan's formula may fail!) but then offers new applications: indeed, contribution (iii) of this paper is concerned with a class of stochastic volatility models introduced by Stein--Stein: we exploit scaling in a way that the small noise asymptotics for the conditioned diffusions gives us a (computable) expression for the asymptotic slope of local variance (the square of Dupire's local volatility as induced from option prices).
Postponing precise assumptions to Section \ref{s:VaradhanSmallNoise}, our first main result is stated as follows

\begin{theorem} \label{t:smallTimeIntro}
(i) Let $X_t = (Y_t, Z_t)$ as above. Under a strong H\"ormander condition, $Y_t$ admits a density $f_t$ for $t>0$ and the following Varadhan type formula holds: for every $y \in \R^l$
\[
\lim_{t \to 0} t \log f_t (y) = -\inf_{\{x=(y,z):z\in\R^{n-l}\}} \Lambda(x)=: \Lambda(N_y)
\]
where $2\Lambda(x)=d^2(x_0,x)$ is the squared control distance associated to $\{\sigma_1,...,\sigma_d\}$ and $N_y = \{x $ $=(y,z):z\in\R^{n-l}\}$.

(ii) Under a further technical assumption (always satisfied in the elliptic case - see Thm \ref{t:conditionalLaw} for a precise statement)
\begin{equation} \label{CondLawT}
\mathcal{L}\left( Z_{t} |Y_{t} = y\right)
\Longrightarrow \delta _{z^{\ast }(y)} \qquad \text{as } t
\downarrow 0
\end{equation}
in the sense of weak convergence of probability measures, provided there exists a unique minimizer $z^{\ast }(y)$ for the problem
$\mbox{argmin}_{x \in N_y} \Lambda(x)$.
\end{theorem}

While the above theorem is clearly useful (it implies, for instance, short time asymptotics for local volatility; on a technical level, we remove the ellipticity requirement from \cite{BBFstochVol2004}), it does not lend itself to understand spatial asymptotics. To this end, we generalize the setup and discuss small noise problems of the form
\be \label{e:baseSDE} %\label{e:SDEmainTheorem}
dX^{\varepsilon}_t=b_{\varepsilon }\left( X^{\varepsilon}_t\right)
dt+\varepsilon \sum_{j=1}^d \sigma_j \left( X^{\varepsilon}_t\right) dW^j_t, \ X_{0}^{\varepsilon
}=x_{0}^{\varepsilon},
\ee
with $x_0^{\eps} \to x_0$ and $b_{\varepsilon }\rightarrow b_0$ (in the precise sense of condition \eqref{e:convergDrift} below).
By Brownian scaling, the short time setting $t\to0$ falls into this setting by taking $t=\varepsilon^2$, $b_{\varepsilon} = \varepsilon^2 b$ and $x_0^{\eps} = x_0$.

\begin{theorem} \label{t:smallNoiseIntro}
(i)
Write $X^{\varepsilon}_t = \left( Y^{\varepsilon }_t,Z^{\varepsilon }_t\right)$. Under a strong H\"ormander condition, and a further technical assumption (which is always satisfied in the elliptic case, or also when $b_0 \equiv 0$ - see Thm \ref{t:VaradhMargin} for precise assumptions), the following Varadhan type formula holds for the density $f^{\eps}_t$ of $Y^{\varepsilon }_t$: for every $t>0$ and $y \in \R^l$
\[
\lim_{\epsilon \to 0} \eps^2 \log f^{\epsilon}_t (y) = -\Lambda_t(N_y).
\]
where, as before $N_y = {\{x=(y,z):z\in\R^{n-l}\}}$. (Although the action $\Lambda$ is still given in terms of a variational problem, cf. \eqref{e:rateFctX} below, it has no more the interpretation as point-to-subspace distance.)

(ii) Under the same assumptions as above
%and again a further technical assumption always satisfied in the elliptic case
we have, for fixed $y$ and $t>0$,

\begin{equation} \label{CondLawEps}
\mathcal{L}\left( Z_{t}^{\varepsilon }|Y_{t}^{\varepsilon }=y\right)
\Longrightarrow \delta _{z_{t}^{\ast }(y)} \qquad \text{as }\varepsilon
\downarrow 0
\end{equation}
provided there exists a unique minimizer $z_{t}^{\ast }(y)$ for the problem
$\mbox{argmin}_{x\in N_y}\Lambda_t(x)$.
\end{theorem}

\noindent
The case $l=n$ in Theorem \ref{t:smallTimeIntro}(ii) (``from $x_0 \in \R^n$ to $x \in \R^n$'') is covered by the result of Molchanov \cite{Mol75} for elliptic diffusions, and more recently by Bailleul \cite{BailleulBridges} in the hypoelliptic setting.
Besides the more general framework of small noise asymptotics that we consider in Theorem \ref{t:smallNoiseIntro}, in (ii) the final target set for the process $X$ is an affine subspace instead of a single point (``from $x_0$ to $N_y$'', restoring the point-to-point situation when $l=n$).
Also, the results of \cite{Mol75,BailleulBridges} are given on compact manifolds, while we work here with $\R^n$-valued processes, and need to rely on some non-trivial tail bounds.

Following the well-known projection results \cite{GyonLocVol, Dupire} for It\^o SDEs, we then have the following corollaries of Theorem \ref{t:smallNoiseIntro} for the local volatility:

\begin{theorem} \label{t:localVolIntro} 
(i) [\textbf{Local volatility, short time behavior}] In a generic stochastic volatility model $(Y,Z)$ (where $Y$ denotes log-price and $Z$ stochastic volatility)
\begin{equation}  \label{e:shortTimeLocalVolIntro}
\sigma^2_{loc}(t,y) = 
\esp[\left(Z_t\right)^2 |Y_t=y] \to z^*(y)^2 \qquad \mbox{as } t \to 0.
\end{equation}
Here $z^*=z^*(y)$ is the ``most likely'' arrival point, computed as $argmin$ of $z \in \R \mapsto \Lambda^{SV}(y,z)$ where $\Lambda^{SV}(\cdot)$ is the action associated to the stochastic volatility model. (Eventually, explicit computations depend on the specific model considered.)

(ii) [\textbf{Local volatility `wings' in the Stein--Stein model}] In the Stein--Stein model (where $Z$ follows an Ornstein--Uhlenbeck process, see equation \eqref{e:steinStein}) \footnote{When the correlation parameter between the log-price $Y$ and the instantaneous volatility $Z$ is not null, the Stein--Stein model is also known as Sch{\"o}bel and Zhu model, see \cite{SchZhu}.}, small noise asymptotics lead to 
\be  \label{e:steinSteinWingsIntro}
\lim_{y \to \pm \infty} \frac{\sigma^2_{\mathrm{loc}}(t,y)}{|y|}
= 
\lim_{y \to \pm \infty} \frac1{|y|}
\esp[\left(Z_t\right)^2 |Y_t=y]
=
c_t^{\pm} > 0
\ee
where the constants $c_t^{\pm}$ are given explicitly in terms of the model parameters. 
\end{theorem}

\noindent
%The result in Theorem \ref{t:localVolIntro}(ii) was announced in \cite{DmFG}.
As we will show in Section \ref{s:localVol}, under some special parameter configuration of the Stein--Stein model, the explicit expression of the constants $c_t^{\pm}$ appearing in \eqref{e:steinSteinWingsIntro} turns out to be consistent with known results from moment explosion for affine models \cite{KellRes}.

As a particular consequence of Theorem \ref{t:localVolIntro}(i), we see that the local volatility surface generated by a fairly general stochastic volatility model (see Theorem \ref{t:smallTimeLocalVol} below for detailed assumptions) \emph{does not} explode in small time. 
Conversely, realistic local volatility surfaces typically do exhibit an explosive behavior out of the money when time goes to zero (namely $\sigma_{\mathrm{loc}}(t,y) \to \infty$ as $t \to 0$, for all $y \neq 0$).
This is the case for the parametric local volatility surface calibrated to SPX option data in \cite[Section 4]{GathWang}, and for the local volatilities obtained via Dupire's formula from realistic SSVI parameterisations \cite{SSVI} of the implied volatility (and finally, the same behavior also appears when applying Dupire's formula to option prices generated by models with jumps - see a related discussion in \cite{FrizGerhYor}).
Apart from providing a tool for computing the limit, Theorem \ref{t:localVolIntro}(i) tells precisely that the local volatilities generated by a standard homogeneous stoch vol model based on Brownian diffusions are not able to capture this phenomenon -- thus providing a negative result for this class of models, in the same spirit as \cite[Theorem 2]{FukasawaSkew}, which focuses on the behavior of the at-the-money implied volatility skew.

On the other side, in analogy with the large-strike behavior of implied volatility \cite{Lee, GulForm}, the linear asymptotic behavior of the local variance in Theorem \ref{t:localVolIntro}(ii) is likely to hold in even wider classes of stochastic volatility models (the same result is indeed known to hold for the Heston model, see \cite{DmFG}, based on affine principles).
On the one hand, the knowledge of an explicit spatial asymptotics for the local volatility can motivate the choice of functional forms used to smooth out and/or extrapolate a local volatility surface calibrated to market data.
Already in use among practitioners, SVI-type parameterizations of the local variance, cf. again \cite[Section 4]{GathWang}, are compatible with the asymptotic result in \eqref{e:steinSteinWingsIntro}.
On the other hand, a robust implementation of the local volatility surface isof course the basis for a Monte-Carlo evaluation of exotic options under local volatility.
Once this step is achieved, the comparison of the prices of volatility-sensitive products (cliquets, barriers,...) under a stochastic volatility model and the corresponding `projected' loc vol model is often used by option trading desks in order to quantify the impact related to different volatility dynamics. This procedure often enters as an important step in the assessment of volatility model risk by model validation teams.
Theorem \ref{t:localVolIntro} allows to extrapolate the local volatility function with explicit formulae in extreme regions, where the implementation of Dupire's formula typically suffers from numerical instabilities.
\medskip

\textbf{Acknowledgment.} SDM thanks Davide Barilari and Luca Rizzi for interesting discussions and insights on affine control systems.
PKF acknowledges support from European Research Council under the European Unions Seventh
Framework Programme (FP7/2007-2013) / ERC Grant Agreement \#258237 and DFG grant FR2943/2.
SDM acknowledges funding from the research programs `Chaire Risques Financiers', `Chaire March\'es en mutation' and `Chaire Finance et d\'eveloppement durable'.

\subsection{Small noise systems}

Standing assumption throughout this paper is that the vector fields $\sigma_1, \dots, \sigma_d$ and the functions of the one parameter family $b_{\eps}$ are smooth ($C^{\infty}$) functions: postponing any precise set of assumptions to the following sections, let us say here that our main results are stated under a boundedness assumption on the $b_{\eps}$ and the $\sigma_j$ together with their derivatives of all orders, and then extended to a class of 2-dimensional diffusions (stochastic volatility models) with unbounded coefficients.
We assume that
\be \label{e:initCond} 
x_0^{\eps} \rightarrow x_0 \in \R^n \qquad \text{as $\eps \downarrow 0$},
\ee
and that $b_{\eps}$ converges to some limit vector field $b_0 \in C^{\infty}$
\be \label{e:convergDriftSimple}
b_{\eps} \to b_0 \qquad \text{as $\eps \downarrow 0$}
\ee
uniformly on compact sets of $\R^n$.

Under assumptions \eqref{e:initCond} and \eqref{e:convergDriftSimple}, it is known that the process $X^{\eps}$ satisfies a Large Deviation Principle (LDP) on the path space $C([0,T];\R^n)$ as $\eps \downarrow 0$ (for a nice recent summary about large deviation principles for small-noise diffusions, see Baldi and Caramellino \cite{BaldiCar}, and references therein).
The deviations of $X^{\eps}$ are driven by the solutions of the limiting controlled differential system
\begin{equation} \label{e:detFlow}
d\varphi^h_t = b_0(\varphi^h_t) dt + \sum_{j=1}^d \sigma_j (\varphi^h_t) d h^j_{t}, \quad \varphi^h_0 = x_0,
\end{equation}
where $h \in H_T \subset C([0,T]; \R^n)$, and for any $t \le T$, $H_t$ denotes the Cameron-Martin (Hilbert) space of absolutely continuous functions with derivative in $L^2([0,t]; \R^n)$, equipped with the norm\footnote{We commit a slight abuse of notation writing $|h|_H$, instead of $|h|_{H_t}$, for $h \in H_t$: the time variable, kept fixed in our results, will always be clear from the context.} $|h|_H^2 := || \dot{h}||_{L^2}^2 = \int_0^t |\dot{h}_s|^2 ds$.
Following the typical terminology in large deviations theory, for every $t \le T$ we define the action function $\Lambda_t: \R^n \to [0,\infty)$ by
\begin{equation} \label{e:rateFctX}
\Lambda_t(x) = \inf \Bigl\{ \frac12 |h|_H^2 : h \in \mathcal{K}_t^x \Bigr\}, \qquad x \in \R^n,
\end{equation}
with the convention $\inf \emptyset=\infty$, where 
\[
\mathcal{K}_t^x = \{ h \in H_t : \varphi^h_t(x_0)=x \}
\]
is the set of controls steering the trajectories of the system \eqref{e:detFlow} from the point $x_0$ to the point $x$ in time $t$.
Following standard terminology, we call \emph{minimizing control} any control $h_0 \in \mathcal{K}_t^x$ realizing the infimum in \eqref{e:rateFctX}, namely such that $\frac12 |h_0|_H^2 = \Lambda_t(x)$.
Some properties of $\Lambda_t$ are presented in Lemma \ref{l:propertiesLambda} below.
For every fixed $t >0$, the LDP for the family of finite dimensional random variables $\{ X^{\eps}_t \}_{\eps}$ reads
\begin{equation} \label{e:LDPlawX}
%\begin{aligned}
\limsup_{\eps \to 0} \eps^2 \log \Prob(X^{\eps}_t \in C) \le -\Lambda_t(C);
%\\
\qquad
\liminf_{\eps \to 0} \eps^2 \log \Prob(X^{\eps}_t \in G) \ge -\Lambda_t(G),
%\end{aligned}
\end{equation}
for every closed set $C$ and open set $G$ in $\R^n$.
Following a common convention in large deviations theory, we denote $\Lambda_t(E)= \inf_{x \in E} \Lambda_t(x)$.
\medskip

The large deviations principle \eqref{e:LDPlawX} is very general, and depends only on some mild Lipschitz conditions on the coefficients of the SDE. 
We will be concerned with the situation where the fixed-time distribution of $X^{\eps}$ possesses a density: as it is common in the field of hypoelliptic heat kernel asymptotics \cite{LeMaj,LeMin,BenA}, we assume that strong \Horm condition holds at all points:
\be \label{e:strongH}
\begin{aligned}
\text{(sH)} \quad Lie(\sigma_1, \dots,\sigma_d)_x := \text{span} \{ &\sigma_1, \dots,\sigma_d ; \ [\sigma_i, \sigma_j] : 1 \le i,j \le d;
\\
&[[\sigma_i, [\sigma_l, \sigma_m]] : 1 \le i,l,m \le d; \dots \} \bigr|_x = \R^n \quad	\forall x \in \R^n,
\end{aligned}
\ee
that is, the linear span of the $\sigma_1, \dots, \sigma_d$ and all their Lie brackets\footnote{By definition, $[\sigma_i, \sigma_j ]^k(x) = \sigma_i \cdot \nabla \sigma_j^k(x) - \sigma_j \cdot \nabla \sigma_i^k(x)$, $k=1,\dots,n$.} is the full tangent space to $\R^n$ at $x$ for all $x \in \R^n$.
It is a classical result (due to \Horm, Malliavin) that the law of $X^{\eps}_t$ admits a smooth density with respect to the Lebesgue measure on $\R^n$ for every $t>0$.\footnote{As is well known, weak \Horm condition at the starting point $x_0$ is a sufficient condition to have a smooth density.
Some of our technical results are actually proved under this assumption (weaker than \eqref{e:strongH}), see Lemma \ref{l:keyEstim} in Appendix \ref{s:app2}.}

In order to study the asymptotic behavior of the density of $X^{\eps}_t$, we impose the convergence of the partial derivatives of the drift vector field $b_{\eps}$: in addition to \eqref{e:convergDriftSimple}, for every multi-index $\alpha \in \{1,\dots,n\}^k$
\be \label{e:convergDrift}
\partial^{\alpha}_x b_{\eps} \to \partial^{\alpha}_x b_0
\qquad \text{as $\eps \downarrow 0$}
\ee
uniformly on compact sets of $\R^n$, where $\partial^{\alpha}_x := \partial^k_{x_{\alpha_1},\dots,x_{\alpha_k}}$.
Furthermore, we assume that the families of norms  $|b_{\eps}|_{\infty}$ and $|\partial^{\alpha}_x b_{\eps}|_{\infty}$ are uniformly bounded in $\eps$, for every $\alpha$.
\medskip

\textbf{The deterministic Malliavin matrix $C_{x_0}(h)$}.
For every $t \in [0,T]$, the map $h \mapsto \varphi^h_t(x_0)$ is differentiable (indeed, $C^{\infty}$) from $H$ into $\R^n$, see Bismut \cite[Theorem 1.1]{Bism}.
Let us denote $D \varphi^h_t(x_0) \in Lin(H,\R^n)$ its Fr\'echet derivative at $h$. 
On the other hand, for fixed $h$, $\varphi^h_t(x)$ is a diffeomorphism as a function of $x \in \R^n$: we denote $\Phi^h_t(x) \in Lin(\R^n,\R^n)$ its differential at $x$.
The method of variation of constants allows to express $D \varphi^h_t(x_0)[k]$, the image of $k \in H$ through the linear map $D \varphi^h_t(x_0)$, via the representation formula
\[
D \varphi^h_t(x_0)[k] =
%\langle D \varphi^h_t(x_0), k \rangle_H =
\int_0^t \sum_{j=1}^d \Phi^h_t(x_0)\Phi^h_s(x_s)^{-1} \sigma_j(x_s) \dot{k}^j_s ds, \qquad x_s=\varphi^h_s(x_0).
\]

Following Bismut \cite{Bism}, and in analogy with the stochastic Malliavin matrix, we introduce the \emph{deterministic} Malliavin covariance matrix $C_{x_0}(h)$, whose entries are given by
\be \label{e:mallMatrix}
\begin{aligned}
C_{x_0}(h)^{i,j} &=
%\langle D \varphi^{h,i}_t(x_0), D \varphi^{h,j}_t(x_0) \rangle_H
%\\ &=
\int_0^t \sum_{l=1}^d
\left[ \Phi^h_t(x_0) \Phi^h_s(x_s)^{-1} \sigma_l(x_s) \right]^i
\left[ \Phi^h_t(x_0) \Phi^h_s(x_s)^{-1} \sigma_l(x_s) \right]^j
ds.
\end{aligned}
\ee
It is a fundamental remark due to Bismut \cite[Theorem 1.3]{Bism} that
%$\varphi_t^h(x_0)$ is a submersion at $h$
$D \varphi^{h}_t(x_0)$ has full rank $n$ if and only if the matrix $C_{x_0}(h)$ is invertible.
The invertibility of $C_{x_0}(h)$ is related to the non-degeneracy of the vector fields $\sigma_j$; in the presence of a locally elliptic diffusion coefficient - which is the case for several financial applications - the following invertibility condition is useful, and easy to check:
\begin{lemma} \label{l:invertibilityLocalEllipt}
Let $h \in \mathcal{K}_t^x$. If there exists $s \in [0,t]$ such that
\[
\mbox{span}[\sigma_1, \dots, \sigma_d]_{x_s} = \R^n, \qquad x_s=\varphi^h_s(x_0),
\]
then $C_{x_0}(h)$ is invertible.
\end{lemma}

\noindent
The proof of Lemma \ref{l:invertibilityLocalEllipt} is an easy linear-algebra exercise, see \cite[Theorem 1.10]{Bism} or \cite[Proposition 2.1]{DFJV:I}.
A sufficient condition for $C_{x_0}(h)$ to be invertible for every $h \neq 0$, stronger than \Horm's condition, is given as condition (H2) in \cite[Chap.1]{Bism}.
\medskip

\textbf{Notation for densities.} We denote
\[
p_t^{\eps}(\cdot) = p_t^{\eps}(x), \qquad x \in \R^n
\]
the density of $X^{\eps}_t$ and
\[
f_t^{\eps}(\cdot) = f_t^{\eps}(y), \qquad y \in \R^l
\]
the density of the $\R^l$-valued projection
\[
Y^{\eps}_t = \Pi_l X^{\eps}_t := (X^{\eps,1}_t, \dots, X^{\eps,l}_t) \in \R^l, \qquad l \le n.
\]
%Since $X^{\eps}_t$ admits a smooth density for all $t>0$, the same is true for the $l$--dimensional projection
%$Y^{\eps}_t$;
It is clear that 
\begin{equation} \label{e:reprMarginal}
f_t^{\eps}(y) = \int_{\R^{n-l}} p_t^{\eps}(y,z) dz,
\end{equation}
where $p_t^{\eps}(y,z):= p_t^{\eps}((y,z))$.
Note that the (limiting) initial condition $x_0$ is fixed in the present discussion and, in contrast with the usual convention in heat kernel analysis, we do not write $p_t(x_0,x)$ - including the initial condition in the symbol for the density - in order to avoid any confusion between initial and terminal points when writing $p_t^{\eps}(y,z)$ for $(y,z) \in \R^{l} \times \R^{n-l} \cong  \R^n$.

Finally, we denote $|\cdot|$ the infinity norm in $\R^n$, and $B_R(x)$ (resp. $B_R^c(x)$) the associated closed ball of radius $R$ around $x$ (resp. the complementary of the ball).

\subsection{Relation to the works of Takanobu--Watanabe}  % \footnote{We thank the associate editor for suggesting this detailed discussion.}
\label{sec:TW}

As is well known, conditional expectation of Wiener-functionals can be
analyzed using Watanabe's pullbacks of delta functions \cite[Ch.V, Sec.9]{IW}.
More specifically, for smooth functions $\phi$
%with regard to (\ref{e:shortTimeLocalVolIntro}), $\phi(z)=z^2$
one has
\begin{equation}
\esp \left[ \phi \left( Z_{1}\right) |Y_{1}=y\right] =\frac{\esp \left[ \delta
_{y}\left( Y_{1}\right) G\left( \omega \right) \right] }{\esp \left[ \delta
_{y}\left( Y_{1}\right) \right] } 
\quad
\text{ with }G\left( \omega \right) =\phi
\left( Z_{1}\right), \label{equ:condEgenW}
\end{equation}%
where $\delta_y(\cdot)$ denotes the Dirac delta function at $y$.
When $X=(Y,Z)$ is the solution to a stochastic differential equation, assuming
small noise dynamics of the form%
\begin{equation}
dX_{t} = \varepsilon ^{2}b\left( X_{t}\right) dt+\varepsilon
\sum_{j=1}^{d}\sigma_{j}\left( X_{t}\right)
dW_{t},\,\,X_{0}=x_{0}  \label{equ:TWdX}
\end{equation}%
(cf. \cite[equation (2.1)]{TakaWat}), writing $\left( Y^{\varepsilon
},Z^{\varepsilon }\right) $ to indicate dependence on $\varepsilon >0$,
asymptotic expansions of the form 
\begin{equation}
\esp\left[ \delta _{y}\left( Y_{1}^{\varepsilon }\right) G\left( \varepsilon
,\omega \right) \right] \sim e^{-\Lambda /\varepsilon ^{2}}\varepsilon
^{-(l+m)}c_{0}  \label{e51TW}
\end{equation}%
were obtained in \cite[Thm. 5.1]{TakaWat}, with $c_0 =c_0 (G)$ given in \cite[Equ. (3.2)]{TakaWat}. It is then tempting to combine 
(\ref{equ:condEgenW}) and (\ref{e51TW}) such as to obtain an asymptotic expansion
of 
$$\esp\left[ \phi \left( Z_{1}^{\varepsilon }\right) |Y_{1}^{\varepsilon }=y\right],
$$
in terms of $c_0 (\phi \left( Z_{1}\right)) \, / \, c_0 (1)$.
There are, however, some (serious) obstacles in proceeding this way, a detailed
discussion of which may also help to put this paper's contribution into context. 
Before going into details, recall validity of \cite[Thm. 5.1]{TakaWat}
hings on the following conditions: (A) strong H\"{o}rmander, (B)
non-degeracy of the deterministic Malliavin matrix (essentially a
control-theoretic condition, trivially satisfied in elliptic situations),
(C)\ finite dimension, say $m$, of the space of minimizing controls, where $m=0$ 
means finitely many minimizers,
(D) a certain non-degeneracy of the action functional $I(h)=\frac12|h|_H^2$.\footnote{For $m=0$: each minimizer must be a non-degenerate minimizer of the action; $m>1$, the null-space of the Hessian is assumed to be compatible with the tangent space of the space of minimizers.} 

\begin{itemize}
\item \textbf{Takanobu--Watanabe conditions are difficult to check.} 
Condition (D) in \cite[Thm. 5.1]{TakaWat} 
was left in a ``raw" infinite-dimensional form which the
authors can just about manage in explicit and simple situations \cite[Sec.7]{TakaWat}
related to L\'evy's area.\footnote{On a related note, and with focus on $m=$0,
providing a {\it finite-dimension} criterion (checkable in terms of Hamiltonian ODEs, 
we called this condition ``non-focality") was a key contribution of \cite{DFJV:I,DFJV:II}.}

\item \textbf{Takanobu--Watanabe conditions may not be satisfied. }An
example related to L\'{e}vy's area, with single minimizer ($m=0$) but where
condition (D) fails, is given in \cite[Sec. 7, (III)$_{4}$]{TakaWat}; see
also \cite{DFJV:I} (where this example is shown to be \textit{focal}).
However, an analysis ``by hand" reveals \cite[(7.7)]{TakaWat} a density expansion which implies the correction large deviation behaviour%
\[
\log f_{1}^{\varepsilon }\sim -\Lambda /\varepsilon ^{2}.
\]%
This provides an examples where \cite[Thm. 5.1]{TakaWat} fails to apply, whereas
our Theorem \ref{t:smallNoiseIntro} works.

\item \textbf{The Takanobu--Watanabe setup does not cover our applications.} The
small noise dynamics (\ref{equ:TWdX}) used in Takanobu--Watanabe do not
allow for general $\varepsilon $-dependence in the drift vector field $b$
and initial data $x_{0}$ (with regard to our notation, only the case $b_{\eps} = \eps^2 b$ is covered by \cite{TakaWat}). This generality, however, is crucial in our discussion of local volatility wings, part (ii) of Theorem \ref{t:localVolIntro}, and introduces some non-trivial complications even at the large-deviation level, as seen in part (i) of Theorem \ref{t:BAL} below.

\item \textbf{Possible gap in Takanobu--Watanabe.} According to the recent preprint \cite{Ina14x}, there was no proof available for 
Theorem 2.1. in \cite{TakaWat}, a large deviation result for pinned diffusions measures, on which \cite[Thm 5.1.]{TakaWat} relies.
(A complete proof, based on rough paths, has then been offered by the author of \cite{Ina14x}.)

\end{itemize}

\section{Theoretical main estimates} \label{s:VaradhanSmallNoise}

Ben Arous and L\'eandre \cite[Section 3]{BAL}, showed that the asymptotics of the logarithm of the density for the
small-noise problem \eqref{e:baseSDE} as $\eps \to 0$ might be governed by a different action function (what they call the ``regular'' action)
defined by
\be \label{e:regularAct}
\Lambda_{R,t}(x) = \inf \Bigl\{ \frac12 |h|_H^2 : h \in \mathcal{K}_t^x,
%\varphi_t^h(x_0) \text{ is a submersion at}
C_{x_0}(h) \mbox{ is invertible}\Bigr\}
\ee
with the convention $\inf \emptyset = \infty$.

\begin{theorem}[\textbf{Ben Arous and L\'eandre \cite{BAL} revisited}] \label{t:BAL}
Consider 
\[
dX^{\varepsilon}=b_{\varepsilon}\left( X^{\varepsilon}_t\right)
dt+\varepsilon \sigma \left( X^{\varepsilon}_t \right) dW, \quad X_{0}^{\varepsilon}=x_{0}^{\varepsilon}
\]
with $x_{0}^{\varepsilon} \rightarrow x_{0}$ as $\eps \to 0$, and $b_{\varepsilon }\rightarrow b$ according to \eqref{e:convergDrift}.
Assume strong H\"{o}rmander condition (sH) at all points, and write $p_{t}^{\varepsilon }\left( x\right) $ for the density of $X^{\eps}_t$.
Then
\begin{itemize}
\item[(i)] the following estimates hold:
\be \label{e:BALupPointwise}
\limsup_{\varepsilon \rightarrow 0}\varepsilon ^{2}\log p_{t}^{\epsilon
}(x)\leq -\Lambda _{t}\left( x\right)  
\ee
and
\begin{equation} \label{e:BALdown}
\liminf_{\varepsilon \rightarrow 0}\varepsilon ^{2}\log p_{t}^{\epsilon
}(x)\geq -\Lambda _{R,t}\left( x\right)  
\end{equation}
for every $x \in \R^n$.
In particular, if $\mathcal{K}_{t}^{x}$ is non-empty for some $x$ and there exists a minimizing control $h_0 \in \mathcal{K}_t^{x}$ such that $C_{x_0}(h_0)$ is invertible, then $\Lambda _{t}\left(
x\right) =\Lambda _{R,t}\left( x\right) <\infty $, so that%
\be \label{e:BALcoincide}
\lim_{\varepsilon \rightarrow 0}\varepsilon ^{2}\log p_{t}^{\epsilon
}(x)=-\Lambda _{t}\left( x\right) .
\ee

\item[(ii)] Assume there exists a minimizing control $h_0 \in \mathcal{K}_t^{\overline x}$ with invertible Malliavin matrix $C_{x_0}(h_0)$.
Then, there exists an open neighborhood $V$ of $\overline x$ such that
\begin{equation} \label{e:BALup}
\limsup_{\varepsilon \rightarrow 0}\varepsilon ^{2}\log p_{t}^{\epsilon
}(x)\leq -\Lambda _{t}\left( x\right)  
\end{equation}
holds uniformly over $x$ in compact sets contained in $V$.
\end{itemize}
\end{theorem}

\begin{proof} The proof is given in Appendix \ref{s:app2}. The statement with $b_{\eps}\equiv b_0$ and $x^{\eps}=x_0$, and
without the uniform convergence in \eqref{e:BALup}, is given in Theorem III.1 in \cite{BAL}.
\end{proof}

\begin{remark} \label{r:invertibC}
\emph{If one assumes existence of a minimizing sequence $h_n \in \mathcal{K}_t^x$ %with $\frac12 |h_n|_H^2 \le \Lambda_{t}(x) + \frac1n$
for \eqref{e:rateFctX} such that $C_{x_0}(h_n)$ is invertible for every $n$, then $\Lambda_{t}(x)=\Lambda _{R,t}\left( x\right)$ immediately follows from the definition of the two actions. %in \eqref{e:rateFctX} and \eqref{e:regularAct}.
Under this assumption, \eqref{e:BALcoincide} holds.
In the end, the condition of invertibility of $C_{x_0}(h)$ for all $h$ in $\mathcal{K}_t^x$ (for some point $x$) will be satisfied in our applications.
}
\end{remark}

\noindent
Some additional comments are in order.

\begin{comment} \label{c:BALresult}
\begin{itemize}
\item[(i)]
\emph{In light of Lemma \ref{l:invertibilityLocalEllipt}, if $\sigma_1, \dots, \sigma_d$ span the whole $\R^n$ at either $x_0$ or $x$,
$C_{x_0}(h)$ is invertible at every $h \in \mathcal{K}_t^x$. 
}

\item[(ii)]
\emph{Even when strong \Horm condition is satisfied at all points, $C_{x_0}(h)$ might fail to be invertible on some $h$.
But if $b_0 \equiv 0$ on a neighborhood of $x_0$, the two actions $\Lambda_t$ and $\Lambda_{R,t}$ coincide (see the discussion in \cite[Section 3]{BAL}, using results from \cite{LeMin,LeFibre}).
In this case, \eqref{e:BALcoincide} holds.
}

\item[(iii)]
\emph{In general, the two actions can be different.
In \cite[Section 1]{BAL} an example on $\R^2$ is given, where strong \Horm condition is satisfied at all points, but the two actions $\Lambda_t$ and $\Lambda_{R,t}$ do not coincide.
As a consequence, the classical Varadhan formula \eqref{e:BALcoincide} does not hold at all points.}
\end{itemize}
\end{comment}

\noindent
The following tail bound will be useful in the proof of our main result in the next section.

\begin{proposition} \label{p:tailInt}
Let $t>0$ and $y \in \R^l$ be fixed.
Under the assumption of Theorem \ref{t:BAL}$(i)$, we have, for every $A>0$, every $\overline{z} \in \R^{n-l}$ and every $k \in \mathbb N$
\[
\limsup_{\epsilon \rightarrow 0}\varepsilon ^{2}
\log \int_{
\{ z\in \mathbb{R}^{n-l}:\left\vert z-\overline{z}\right\vert \geq A\}
}
|z|^k p_{t}^{\varepsilon }\left( y,z\right)
dz\leq -\Lambda _{t}(\left\{ \left( y,z\right)
%\in \mathbb{R}^{l}\times  \mathbb{R}^{n-l}
: \left\vert z-\overline{z}\right\vert \geq A\right\}
),
\]%
where with the usual convention, $\Lambda _{t}\left( E\right) =\inf_{x\in E}\Lambda _{t}\left( x\right) $. 
\end{proposition}

\begin{proof}
Given in Appendix \ref{s:app2}.
\end{proof}
\medskip

\noindent
Crucial for the applications, the optimal control problem \eqref{e:rateFctX} defining the action function $\Lambda_t$ can be rephrased in terms of the Hamiltonian formalism.
The following proposition provides necessary optimality conditions for the controls in $\mathcal{K}_t^{(y, \cdot)}$ when $y$ is fixed, in the spirit of Pontryagin's maximum principle: as such, it appears as a generalization of the corresponding result in Bismut \cite{Bism}, from a point-to-point setting ($x_0 \in \R^n$ to $x \in \R^n$) to a point-to-subspace ($x_0 \in \R^n$ to $N_y := (y,\cdot)$, $y \in \R^l$) setting.
Let us introduce the Hamiltonian
\[
\mathcal{H}(x,p) = \langle b_0(x), p \rangle_{\R^n} + \frac12 \sum_{j=1}^d \langle \sigma_j(x), p \rangle^2_{\R^n}.
\]

\begin{proposition}[see Proposition 2 in \cite{DFJV:I}] \label{p:pmp}
Fix $y \in \R^l$, and assume $h_0 \in \mathcal{K}_t^{(y, \cdot)}$ is an optimal control for the problem
\[
\Lambda_t(N_y) = \inf \Bigl\{ \frac12 |h|_H^2 : h \in \mathcal{K}_t^{(y,\cdot)} \Bigr\} \qquad N_y=(y, \cdot).
\]
Moreover, assume the deterministic Malliavin matrix $C_{x_0}(h_0)$ is invertible.
Then, there exists a unique $\overline{p}_0$ such that $\varphi^{h_0}_s(x_0) = x_s$ for all $s \in [0,t]$, where $(x_s,p_s)_{s \le t}$ solves the Hamiltonian ODEs
\be \label{e:hamiltonODE}
\left(\begin{array}{c} \dot{x}_s \\ \dot{p}_s \end{array} \right)
= 
\left(\begin{array}{c} \partial_p\mathcal{H}(x_s,p_s)\\ -\partial_x \mathcal{H}(x_s,p_s) \end{array} \right)
\ee
subject to the (initial-, terminal- and transversality-) boundary conditions
\begin{eqnarray}  \notag
x_0=x_0 \in \R^n, & x_t=(y,\cdot) \in \R^l \times \R^{n-l}
\\ \label{e:transvers}
p_0 = \overline{p}_0 \in \R^n, & p_t=(\cdot,0) \in \R^l \times \R^{n-l}.
\end{eqnarray}
Furthermore, the control $h_0$ is restored as
\[
\dot{h}^j_0(s) = \langle \sigma_j(x_s), p_s \rangle, \quad j = 1, \dots, d
\]
and $\Lambda_t(N_y)=\frac12|h_0|_H^2$.
\end{proposition}

\begin{remark} \label{r:charMinimizer}
\emph{If $x_t=(y, z_t)$ is the terminal value of the $x$-component of a solution to \eqref{e:hamiltonODE}, then $z_t$ is a minimizer of the map $z \mapsto \Lambda_t(y, \cdot)$.}
\end{remark}
\medskip

\noindent
Finally, the following lemma summarizes some properties of the control system \eqref{e:detFlow} and of the action $\Lambda_t$ that will be extensively used throughout the paper.

\begin{lemma} \label{l:propertiesLambda}
Assume the vector fields $b_0$ and $(\sigma_j)_j$ are Lipschitz continuous, with Lipschitz constant $K$.
Denote
$\varphi^h_{\cdot}(x_0)$ the solution of the ODE \eqref{e:detFlow} on the interval $[0,t]$, and $\Lambda_t$ the action of the system as in \eqref{e:rateFctX}.
Then
\begin{enumerate}
\item[(i)] For every $t$, the map $h \to \varphi_t^h(x_0)$ is weakly continuous from $H$ into $\R^n$.
Moreover, there exists a positive constant $C(t,x_0,K)$, increasing in $K$, such that
\be \label{e:gronwallSolutionMap}
\sup_{s \le t} |\varphi^h_s(x_0)| \le C(t,x_0,K) \: e^{C(t,x_0,K) |h|_H}
\ee
for every $h\in H$.

\item[(ii)] $\Lambda_t$ is a good rate function of large deviations theory: that is, for every $l \ge0$ the level sets $\{x: \Lambda_t(x) \le l \}$ are compact.
In particular, $\Lambda_t$ is lower semi-continuous.

\item[(iii)] if $\K_{t}^x \neq \emptyset$, the infimum in \eqref{e:rateFctX} is attained: that is, there exists a minimizing control $h_0 \in \K_{t}^x$ such that $\Lambda_t(x)=\frac12 |h_0|_H^2$.
%We denote $K^x_{\min}=\{ h_0 \in \K_{t}^x: \frac12 |h_0|_H^2=\Lambda_t(x) \}$ the set of minimizing controls.
\end{enumerate}

Assume moreover that the vector fields are $C^{\infty}_b$ (bounded with bounded derivatives). Then

\begin{enumerate}
\item[(iv)] If the $(\sigma_j)_j$ satisfy the strong \Horm condition (sH) at all $x$, then $\K_{t}^x \neq \emptyset$ for every $x$, and $x \mapsto \Lambda_t(x)$ is finite on $\R^n$.

\item[(v)] If there exists a minimizing control $h_0 \in \K_{t}^{\overline{x}}$, $\overline{x} \in \R^n$, with invertible Malliavin matrix $C_{x_0}(h_0)$, then there exists a neighborhood $V$ of $\overline{x}$ such that $\Lambda_t$ is continuous on $V$.
\end{enumerate}
\end{lemma}

\begin{proof}
$(i)$. Weak continuity with respect to the control parameter is classical.
See Bismut \cite[Theorem 1.1]{Bism} for the case of smooth vector fields: the case of Lipschitz continuous coefficients is
handled analogously: in essence, the continuity property and estimate \eqref{e:gronwallSolutionMap} follow from an application of Gronwall's lemma.
See also \cite[proof of Lemma 2.5]{BaldiCar} for estimate \eqref{e:gronwallSolutionMap}.
$(ii)$ is a direct consequence of $(i)$: use $\{x:\Lambda_t(x) \le l \} = \{\varphi^h_t(x_0): |h|^2_H \le 2 l \}$, and the latter set is compact since (weakly) continuous image of a (weakly) compact set.
$(iii)$ is a direct consequence of $(i)$: indeed, assume $\hat{h} \in \mathcal{K}_{t}^x$.
Then the infimum in \eqref{e:rateFctX} is in fact taken over the set $\{h: \varphi_t^h(x_0)=x, |h|_H \le |\hat{h}|_H \}$, which is weakly compact; since the norm $|\cdot|_H$ is weakly lower semi continuous, $\frac12 |h|_H^2$ attains its minimum on this set.
$(iv)$. The non-emptiness of $K_{t}^x$ (therefore, finiteness of $\Lambda_t$) under strong \Horm condition is a classical result of controllability: see e.g. \cite[Theorem 2, p. 106]{Jurdjevic97} for the affine control system with drift that we consider here.
$(v)$. In light of $(ii)$, it is sufficient to prove that $\Lambda_t$ is upper semi-continuous.
Under the assumption of existence of a minimizing control $h_0$ with invertible Malliavin matrix, upper semi-continuity is proven as in the second part of \cite[Proposition 3.2]{CannRiff}: as it is typical in the geometrical control setting, the key point is the implementation of the Implicit Function Theorem locally around $\overline{x}$, which is made possible by the fact that the linear map $D\varphi^{h_0}_t(x_0): H \to \R^n$ has full rank.
\end{proof}

\begin{remark}[\emph{On point $(v)$ of Lemma \ref{l:propertiesLambda}}] \label{r:noDriftLambdaContinuous}
\emph{
When $b_0\equiv 0$ in \eqref{e:detFlow} and strong \Horm condition (sH) holds, it is classical that $x \mapsto \Lambda_t(x)$ is finite and continuous on $\R^n$, without any further assumption about the existence of minimizers with invertible Malliavin matrix.
This statement is equivalent to well-known continuity of the Carnot-Carath\'eodory distance on a sub-Riemannian manifold (here: $\R^n$ equipped with the control distance induced by \eqref{e:detFlow} with $b_0 \equiv 0$).
A standard proof, based on the small-time local controllability of driftless control systems, is provided for example in Bismut \cite[Theorem 1.14]{Bism}.
For affine control systems with non-zero drift as \eqref{e:detFlow}, the continuity of $\Lambda_t$ is not, in general, a consequence of \Horm condition.
In \cite[Section 2]{AgrachevLeeContinuity} an example is provided, where strong \Horm condition holds at all points, and the function $\Lambda_t$ fails to be continuous.
}
\end{remark}

\begin{remark}
\emph{
If the function $\Lambda_t$ is known to be continuous on $\R^n$, estimate \eqref{e:BALup} in Theorem \ref{t:BAL} holds uniformly over $x$ in compact sets of $\R^n$, without any further assumption.
}
\end{remark}

\subsection{The conditioned diffusion} \label{s:conditional}

Denote $Z^{\eps}_t:= (X^{\eps,l+1}_t, \dots, X^{\eps,n}_t)$ the projection of $X^{\eps}_t$ over the last $n-l$ components, 
so that
\[ 
X^{\eps}_t = (Y^{\eps}_t, Z^{\eps}_t).
\]
We write
\[
\mathcal{L} \left( Z_t^{\varepsilon} | Y_t^{\varepsilon}=y \right)
\]
for the law of $Z^{\eps}_t$ conditional on $Y_t^{\varepsilon}$ being at level $y \in \R^l$ at time $t$.
If $f^{\eps}_t(y) > 0$, this is well-defined via
\[
\esp[\varphi(Z^{\eps}_t) | Y_t^{\eps}=y ] = \int_{R^{n-l}} \varphi(z) g^{\eps}(z) dz
\]
for all $\varphi \in C_b(\R^{n-l})$, where 
\begin{equation} \label{e:conditDens}
g^{\eps}(z) = g^{\eps}_{t,y}(z) := \frac{ p^{\eps}_t(y,z) }{ f^{\eps}_t(y) }
\end{equation}
is the density of $Z_t^{\varepsilon}$ conditional on $Y_t^{\varepsilon}=y $.

\begin{theorem} \label{t:conditionalLaw}
Consider $X^{\varepsilon }=\left( Y^{\varepsilon},Z^{\varepsilon }\right) \in \R^l \times \R^{n-l} \cong \R^n$
given by 
\be \label{e:SDEmainTheorem}
dX^{\varepsilon}_t=b_{\varepsilon }\left( X^{\varepsilon}_t\right)
dt+\varepsilon \sum_{j=1}^d \sigma_j \left( X^{\varepsilon}_t\right) dW^j_t, \ X_{0}^{\varepsilon
}=x_{0}^{\varepsilon},
\ee
with $x_0^{\eps} \to x_0$ and $b_{\varepsilon }\rightarrow b_0$ according to \eqref{e:convergDrift}, and assume strong H\"{o}rmander condition (sH) at all points.
Fix $y\in \mathbb{R}^{l}$ and $t>0$, and set $N_y=(y,\cdot)$.
Assume that there exists a unique minimizer $z^*=z_{t}^{\ast}(y)$ for the problem\footnote{The
existence of a minimizer follows from the lower semi-continuity and compactness of the level sets of the map $z \mapsto \Lambda_t (y,z)$.
}
\be \label{e:argMin}
(y,z^*) :=
\mbox{argmin}_{x\in N_y}\Lambda
_{t}(x),
\ee
and assume that for every $z$ in a neighbourhood of $z_{t}^{\ast}(y)$
%in $\R^{n-l}$
there exists a minimizing control $h_0 \in \mathcal{K}_{t}^{(y,z)}$ with invertible deterministic Malliavin matrix $C_{x_{0}}(h_0)$, as defined in \eqref{e:mallMatrix}.
Then, $f^{\eps}_t(y) > 0$ and
\[
\mathcal{L}\left( Z_{t}^{\varepsilon }|Y_{t}^{\varepsilon }=y\right)
\Longrightarrow \delta _{z_{t}^{\ast }(y)} \qquad \text{as }\varepsilon
\downarrow 0
\]
in the sense of weak convergence of probability measures on $\mathbb{R}^{n-l}
$, i.e. for all $\phi \in C_{b}( \mathbb{R}^{n-l}) $,%
\be \label{e:phiConv}
E\left[ \phi \left( Z_{t}^{\varepsilon }\right) |Y_{t}^{\varepsilon }=y%
\right] \rightarrow \phi \left( z_{t}^{\ast }(y)\right) \qquad \text{ as }%
\varepsilon \downarrow 0.
\ee
\end{theorem}
\medskip

\begin{corollary}[\textbf{Test functions with polynomial growth}] \label{c:polGrowth}
Under the assumption of Theorem \ref{t:conditionalLaw}, assume $\phi$ is continuous and has polynomial growth, that is $\phi(z) \le C(1+|z|^k)$ for some $C>0$ and $k \in \mathbb{N}$, for all $z$.
Then
\[
E\left[ \phi \left( Z_{t}^{\varepsilon }\right) |Y_{t}^{\varepsilon }=y%
\right] \rightarrow \phi \left( z_{t}^{\ast }(y)\right) \text{ as }%
\varepsilon \downarrow 0.
\]
holds.
\end{corollary}

Theorem \ref{t:conditionalLaw} and Corollary \ref{c:polGrowth} are proven at the end of this section.

\begin{remark}[\emph{Extension to finitely many argmin's}] \label{r:finitelyMany}
\emph{If there exist finitely many global minimizer $z^{*,i}=z^{*,i}_t(y)$, $i=1,\dots,N$, for the problem \eqref{e:argMin}, assuming that $C_{x_0}(h)$ is invertible for some minimizing control $h_0 \in \mathcal{K}_{t}^{(y,z)}$ for every $z$ in a neighborhood of each $z^{\ast,i}$, a modification of the arguments used in the proof of Theorem \ref{t:conditionalLaw} allows to show that $\mathcal{L}\left( Z_{t}^{\varepsilon }|Y_{t}^{\varepsilon }=y\right)$ converges to a law supported by the $z^{*,i}$, i.e.
\[
\mathcal{L}\left( Z_{t}^{\varepsilon }|Y_{t}^{\varepsilon }=y\right)
\Longrightarrow \sum_{i=1}^N \alpha_i \delta_{z^{\ast,i}} \qquad \text{as }\varepsilon \downarrow 0,
\]
for some $(\alpha_i)_{i=1}^N$ with $\alpha_i \ge 0$ and $\sum_{i=1}^N \alpha_i=1$, which means
$
E\left[ \phi \left( Z_{t}^{\varepsilon }\right) |Y_{t}^{\varepsilon }=y%
\right] \rightarrow \sum_{i=1}^N \alpha_i \phi(z^{\ast,i})$ as $\varepsilon \downarrow 0$, for every $\phi \in C_b$.
}
\end{remark}

\begin{remark}[\emph{Extension to the finite dimensional law}] \label{r:finiteDim}
\emph{Under the hypotheses of Theorem 	\ref{t:conditionalLaw}, assume in addition that there exist a unique minimizing control $h_0$ in $\mathcal{K}^{(y,\cdot)}$, that is
\be \label{e:uniqueControl}
\mathcal{K}^{(y,\cdot)}_{min} := \Bigl\{ h \in \mathcal{K}^{(y,\cdot)} : \frac12|h^2|_H = \Lambda_t(N_y) \Bigr\}
= \{ h_0 \}
\ee
(in particular, $(y,z^*_t(y))=\varphi^{h_0}_t(x_0)$ is the unique minimizer of $\Lambda_t$ on the set $N_y$).
Then, for every $0\le t_1 < \dots < t_n \le t$,
\[
\mathcal{L}(Z^{\eps}_{t_1}, \dots, Z^{\eps}_{t_n}|Y^{\eps}_t = y)
\Rightarrow 
%\otimes_{j=1}^n \delta_{z_{t_j}}
\delta_{(z_{t_1},\dots,z_{t_n})}
\]
where $(y_s,z_s):=\varphi^{h_0}_s(x_0)$, $s \le t$, is the trajectory associated to the control $h_0$.
The case of finitely many minimizing controls $h_0^i$ in $\mathcal{K}^{(y,\cdot)}_{min}$ gives rise to a limiting law supported by the $(z^i_{t_1},\dots,z^i_{t_n})$, with $(y^i_s, z^i_s)=\varphi^{h_0^i}_s(x_0)$.
Subject to a tightness estimate, the convergence of the finite dimensional law yields the convergence at the path level, namely $\mathcal{L}(Z^{\eps}_{\cdot}|Y^{\eps}_t = y) \Rightarrow \delta_{z_{\cdot}}$ in the case of a unique minimizing path $(y_{\cdot},z_{\cdot})=\varphi^{h_0}_{\cdot}(x_0)$.
In the point-to-point case $l=n$, this result is proved in Molchanov \cite{Mol75} for elliptic diffusions and in Bailleul \cite{BailleulBridges} in the hypoelliptic setting (see also Bailleul, Mesnager and Norris \cite{BMNbridges}).
}
\end{remark}

\begin{remark}[\emph{Relation with Takanobu--Watanabe \cite{TakaWat}}] \label{r:comparisonTakaWat}
\emph{As discussed in section \ref{sec:TW}, the asymptotic expansion $\esp \left[ \delta_y(Y_{1}^{\varepsilon}) G(\epsilon, \omega) \right] \sim e^{-\Lambda(N_y)/\eps^2} \eps^{-l+m} c_0$ is given in \cite{TakaWat} under suitable assumptions on $Y_1$ and $
G$, allowing to study
\be \label{eq:condExpRatio}
\esp \left[ \phi \left( Z_{1}^\eps\right) |Y_{1}^\eps=y\right] = \frac{\esp \left[ \delta
_{y}\left( Y_{1}^\eps \right) \phi\left(Z_1^\eps \right) \right] }{\esp \left[ \delta
_{y}\left( Y_{1}^\eps \right) \right]}.
\ee
The constant $c_0$ in the expansion depends on both $Y$ and $G$, see \cite[Theorem 5.1]{TakaWat}.
Let us apply Takanobu--Watanabe's expansion to both the numerator and the denominator at the RHS of \eqref{eq:condExpRatio}, under the assumption of Theorem \ref{t:conditionalLaw} that the minimizer $z^*_1(y)$ is unique.
An inspection of the expression of $c_0$ in \cite[Theorem 5.1]{TakaWat} reveals that the constant at the numerator is given by $c_0 = \tilde c_0 \: \phi(z^*_1(y))$ for some $\tilde c_0 > 0$, while the constant at the denominator is $c_0 = \tilde c_0$.
This means that overall
\[
\esp \left[ \phi \left( Z_{1}^\eps\right) |Y_{1}^\eps=y\right] \to \phi(z^*_1(y)),
\quad \mbox{as } \eps \downarrow 0,
\]
in agreement with Theorem \ref{t:conditionalLaw}.
As pointed out in section \ref{sec:TW}, the validity of Takanobu--Watanabe's expansion relies on some conditions that are cumbersome to check (notably their infinite-dimensional assumption (D)), or may fail be satisfied in situations where Theorem \ref{t:conditionalLaw} is instead applicable.
Moreover, we stress once again that only the case with zero limiting drift $b_0 = \lim_{\eps \to 0} \eps^2 b \equiv 0$ is considered in \cite{TakaWat} - while it is crucial for our applications to space asymptotics of local volatilities in Section \ref{s:asympSlopes} to overcome this limitation.
\\
Of course, using the full expansion $\esp \left[ \delta_y(Y_{1}^{\varepsilon}) G(\epsilon, \omega) \right] = e^{-\Lambda(N_y)/\eps^2} \eps^{-l+m} (c_0 + c_1 \eps + ...)$ in \cite{TakaWat}, higher order expansions in \eqref{eq:condExpRatio} (of the form $\esp \left[ \phi \left( Z_{1}^\eps\right) |Y_{1}^\eps=y\right] = \phi(z^*_1(y)) + c_1(y,\phi) \eps + ...$) are a priori accessible. However, in order to compute the correction $c_1(y,\phi)$, one needs to know the precise value of the constant $\tilde c_0$ for which there is no explicit expression, and significant additional analysis would be be necessary, e.g. in the spirit of \cite{Kusuoka20082545}.
}
\end{remark}

In order to prove Theorem \ref{t:conditionalLaw}, we need a preliminary estimate on the marginal density $f^{\eps}_t$.

\begin{proposition} \label{p:marginVaradh}
Under the hypotheses of Theorem \ref{t:conditionalLaw},
\be \label{e:marginVaradh}
\liminf_{\epsilon \to 0} \eps^2 \log f^{\epsilon}_t (y) \ge -\Lambda_t(N_y),
\ee
where with the usual convention, $\Lambda_t(E) = \inf_{x \in E} \Lambda_t(x)$.
In particular, $f^{\epsilon}_t(y) > 0$ for $\eps$ small enough.\footnote{The strict positivity of $f^{\epsilon}_t$ is not, in general, a consequence of \Horm condition.
According to Ben Arous and L\'eandre's support theorem \cite[Theorem II.1]{BAL}, the density $p_t^{\eps}(x)$ of the full process $X^{\eps}_t$ is strictly positive at $x$ if and only if there exists some $h \in \mathcal{K}_t^x$ such that $C_{x_0}(h)$ is invertible.
While \Horm condition ensures that $\mathcal{K}_t^x$ is non-empty for every $x$, the lack of controls with invertible Malliavin matrix can not be excluded a priori.}
\end{proposition}

\begin{proof}
For simplicity, let us drop the explicit dependence on the (fixed) $t>0$, and write $\Lambda$ for $\Lambda_t$,
$p^{\eps}$ for $p^{\eps}_t$, etc.
Also write $z^*$ for $z^*_t(y)$.
We note that by definition, $\Lambda(N_y) = \inf_z \Lambda_t(y,z) = \Lambda_t(y,z^*)$.
Using Ben Arous and L\'eandre's support theorem \cite[Theorem II.1]{BAL}, the invertibility of $C_{x_{0}}(h)$ for some $h\in \mathcal{K}_{t}^{(y,z^*)}$ implies $p^{\eps}(y,z^*)>0$ for $\eps$ small enough: it follows that $f^{\eps}(y) > 0$, therefore the conditional density $g^{\eps}$ in \eqref{e:conditDens} is well-defined.

Let $K$ be a neighborhood of $z^*$ in $\R^{n-l}$ such that, for every $z \in K$, $C_{x_{0}}(h_0)$ is invertible for some minimizing control $h_0 \in \mathcal{K}_{t}^{(y,z)}$ (with no loss of generality, we may assume $K$ to be compact).
It follows from	Theorem \ref{t:BAL}$(i)$ that
\be \label{e:actionsCoincide} 
\Lambda(y,z)=\Lambda_R(y,z) < \infty, \qquad \forall z \in K.
\ee
From point $(v)$ in Lemma \ref{l:propertiesLambda}, $\Lambda$ is continuous on a neighborhood of $(y,z^*)$: possibly making $K$ smaller, we can assume
\[
\Lambda(y,z) - \Lambda(y,z^*) \le \delta, \qquad \forall z \in K,
\]
for some fixed $\delta > 0$.
It follows from estimate \eqref{e:BALupPointwise} in Theorem \ref{t:BAL} that there exists $\eps_0= \eps_0(\delta)$ such that
\[
p^{\eps}(y,z^*) \le
\exp \Bigl( \frac{-\Lambda(y,z^*)+\delta}{\eps^2} \Bigr)
\]
for every $\eps<\eps_0$.
Analogously, for every $z \in K$ there exists $\eps(z) = \eps(z,\delta)$ such that
\[
p^{\eps}(y,z) \ge
\exp \Bigl( \frac{-\Lambda(y,z)-\delta}{\eps^2} \Bigr)
\]
for all $\eps < \eps(z)$.
It follows from these last two estimates that for every $z \in K$,
\be \label{e:ratio}
\frac{p^{\eps}(y,z)}{p^{\eps}(y,z^*)} \ge
\exp \Bigl( \frac{-(\Lambda(y,z)-\Lambda(y,z^*))-2\delta}{\eps^2} \Bigr)
\ge \exp \Bigl( -\frac{3\delta}{\eps^2} \Bigr)
\ee
for all $\eps < \eps_0 \wedge \eps(z)$.
Now write
\be \label{e:factorOut}
\begin{aligned}
f^{\eps}(y) = \int_{R^{n-l}} p^{\eps}(y,z) dz &\ge \int_K p^{\eps}(y,z) dz
\\
&=p^{\eps}(y,z^*) \: \lambda^{n-l}(K) \int_{K} \frac{p^{\eps}(y,z)}{p^{\eps}(y,z^*)} \frac{dz}{\lambda^{n-l}(K)}.
\end{aligned}
\ee
where $\lambda^{n-l}$ is the Lebesgue measure on $\R^{n-l}$.
First applying Jensen's inequality, then Fatou's lemma, one has
\be \label{e:lowerBoundIntegral}
\begin{aligned}
\liminf_{\eps \to 0} \eps^2 \log
\int_{K} \frac{p^{\eps}(y,z)}{p^{\eps}(y,z^*)} \frac{dz}{\lambda^{n-l}(K)}
&\ge
\liminf_{\eps \to 0}
\int_{K} \eps^2 \log \frac{p^{\eps}(y,z)}{p^{\eps}(y,z^*)} \frac{dz}{\lambda^{n-l}(K)}
\\
&\ge 
\int_{K} \liminf_{\eps \to 0}
\left[ \eps^2 \log \frac{p^{\eps}(y,z)}{p^{\eps}(y,z^*)} \right] \frac{dz}{\lambda^{n-l}(K)}
\\
&\ge
-3\delta.
\end{aligned}
\ee
where we have used \eqref{e:ratio} in the last inequality.
Finally using the lower bound \eqref{e:BALdown} for $p^{\eps}(y,z^*)$ and \eqref{e:lowerBoundIntegral}, it follows from \eqref{e:factorOut} that
\[
\begin{aligned}
\liminf_{\eps \to 0} \eps^2 \log f^{\eps}(y) &\ge
\liminf_{\eps \to 0} \eps^2 \log p^{\eps}(y,z^*)
+
\liminf_{\eps \to 0} \eps^2 \log \left( \lambda^{n-l}(K) \int_{K} \frac{p^{\eps}(y,z)}{p_t^{\eps}(y,z^*)}
\frac{dz}{\lambda^{n-l}(K)} \right)
\\
&\ge -\Lambda(y,z^*) - 3\delta.
\end{aligned}
\]
Since $\delta$ was arbitrary, \eqref{e:marginVaradh} is proved.
%The last statement on the positivity of $f^{\eps}_t(y)$ is a simple consequence of \eqref{e:marginVaradh} and $\Lambda_t(N_y) < \infty$.
\end{proof}

\begin{remark} \label{r:remarkVaradMarginal} \label{r:lowerBoundMinimizerNotUnique}
\emph{Note that we have nowhere used, in the proof of Proposition \ref{p:marginVaradh}, the fact that the minimizer $z^*_t(y)$ is unique, neither that there exist finitely many.
Proposition \ref{p:marginVaradh} then holds under the weaker assumption that $C_{x_0}(h_0)$ is invertible on at least one minimizing control $h_0 \in \mathcal{K}_t^{(y,z)}$, for every $z$ in a neighborhood of \emph{some} global minimizer $z^*$ of the map
$z \mapsto \Lambda_t(y,z)$.
}
\end{remark}

\begin{proof}[Proof of Theorem \ref{t:conditionalLaw}] Let us drop the fixed $t$ from the notation, and write $\Lambda$ for $\Lambda_t$, etc.

\emph{Step 1}.
We want to show that for all $\phi \in C_b$, one has $\bigl| \int_{R^{n-l}} \phi(z) g^{\eps}(z) dz - \phi(z^*) \bigr| \to 0$ as $\eps \to 0$, with $z^* := z^*_t(y)$.
Consider $\eta^{(2)}>0$ such that the map $z \mapsto \Lambda(y,z)$ is continuous on $B_{\eta^{(2)}}(z^*)$ (by Lemma \ref{l:propertiesLambda}, such a $\eta^{(2)}$ exists) and such that estimate \eqref{e:BALup} of Theorem \ref{t:BAL} holds on $B_{\eta^{(2)}}(z^*)$.
Let $\delta > 0$, and consider $\eta^{(1)}=\eta^{(1)}_{\delta}$ such that $0<\eta^{(1)}<\eta^{(2)}$ and $Osc(\phi,\eta^{(1)}) \le \delta$, where $Osc(\phi,\eta) =\sup\{|\phi(z_1)-\phi(z_1)|:|z_1-z_2|\le \eta\}$.
We have
\[
\begin{aligned}
\Bigl| \int_{R^{n-l}} \phi(z) g^{\eps}(z) dz - \varphi(z^*) \Bigr|
&\le
\int_{R^{n-l}} | \phi(z) - \phi(z^*)| g^{\eps}(z) dz
\\
&=
\int_{|z-z^*|\le \eta^{(1)}} | \phi(z) - \phi(z^*)| g^{\eps}(z) dz
+
\int_{|z-z^*|> \eta^{(1)}} | \phi(z) - \phi(z^*)| g^{\eps}(z) dz
\\
&\le \delta
+ 2 |\phi|_{\infty} \int_{|z-z^*|> \eta^{(1)}} g^{\eps}(z) dz,
\end{aligned}
\]
and our aim is to show that the last integral converges to $0$ as $\eps \to 0$.

We have
\[
\int_{|z-z^*|> \eta^{(1)}} g^{\eps}(z) dz
\le
\int_{ \eta^{(1)} < |z-z^*| < \eta^{(2)} } g^{\eps}(z) dz
+
\int_{ |z - z^*| \ge \eta^{(2)} } g^{\eps}(z) dz
:= I^{\eps}_1 + I^{\eps}_2.
\]

\emph{Step 2 ($I^{\eps}_1 \to 0$)}.
Set
\[
a_{\delta} := \inf  \{ \Lambda(y,z) - \Lambda(y,z^*) : |z-z^*| \ge \eta^{(1)} \}.
\]
By the lower semi-continuity of $z \mapsto \Lambda(y,z)$ and the uniqueness of the minimizer $z^*$,
one has $a_{\delta} > 0$.
Let now $\delta_1$ be such that $0 < \delta_1 < a_{\delta}/4$: on the one hand, since estimate \eqref{e:BALup} in Theorem \ref{t:BAL} is uniform over a compact neighborhood of $z^*$, we know there exists $\eps_0 = \eps_0(\delta_1,\eta^{(2)})$ such that
\[
p^{\eps}(y,z) \le \exp \Bigl( \frac{-\Lambda(y,z) + \delta_1}{\eps^2} \Bigr)
\]
for all $z \in B_{\eta^{(2)}}(z^*)$ and $\eps < \eps_0$.
On the other hand, it follows from estimate \eqref{e:marginVaradh} in Proposition \ref{p:marginVaradh} that
\[
f^{\eps}(y) \ge \exp \Bigl( \frac{-\Lambda(y,z^*) - \delta_1}{\eps^2} \Bigr)
\]
for all $\eps < \eps_0$.
Putting these two estimates together and using the definition of $a_{\delta}$, it follows
\begin{equation} \label{e:condComp}
g^{\eps}(z) \le
\exp \Bigl( \frac{2\delta_1-(\Lambda(y,z)-\Lambda(y,z^*)) }{\eps^2} \Bigr)
\le 
\exp \Bigl( -\frac{a_{\delta}}{2\eps^2} \Bigr)
\end{equation}
for all $z$ such that $\eta^{(1)} < |z-z^*| < \eta^{(2)}$ and $\eps < \eps_0$.
Therefore, we have
\[
I^{\eps}_1 \le \exp \Bigl( -\frac{a_{\delta}}{2\eps^2} \Bigr) \: \lambda^{n-l}( B_{\eta^{(2)}}(z^*) ),
\]
for all $\eps < \eps_0$, where $\lambda^{n-l}$ is the Lebesgue measure on $\R^{n-l}$.
For every choice of $\delta$ and $\eta^{(2)}$, the right hand side can be made arbitrarily small taking $\eps$
small enough.

\emph{Step 3 ($I^{\eps}_2 \to 0$)}. 
As done in Step 2, notice that $a^{(2)}:=\inf  \{\Lambda(y,z)-\Lambda(y,z^*):|z-z^*| \ge \eta^{(2)}\} > 0$.
Since 
\[
I^{\eps}_2 = \int_{ |z-z^*|\ge \eta^{(2)} } g^{\eps}(z) dz =
\frac 1{f^{\eps}(y)} \int_{|z-z^*| \ge \eta^{(2)}} p^{\eps}(y,z) dz,
\]
it follows from \ref{p:marginVaradh} and Proposition \ref{p:tailInt} that
\[
\begin{aligned}
\limsup_{\eps \to 0} \eps^2 \log I^{\eps}_2 &\le
- \liminf_{\eps \to 0} \eps^2 \log f^{\eps}(y)
+
\limsup_{\eps \to 0} \eps^2 \log \int_{ |z-z^*| \ge \eta^{(2)}} p^{\eps}(y,z) dz
\\
&\le
\Lambda(y,z^*) - \inf \{ \Lambda(y,z): |z-z^*| \ge \eta^{(2)} \}
\\
&\le \Lambda(y,z^*)-\Lambda(y,z^*)-a^{(2)} = -a^{(2)}<0.
\end{aligned}
\]
The last inequality clearly implies that $I^{\eps}_2$ vanishes as $\eps \to 0$.
\end{proof}
\medskip

\begin{proof}[Proof of Corollary \ref{c:polGrowth}]
Let $\phi_R$, $R>0$, be a bounded continuous function that coincides with $\phi$ on ${B_R(0)}$.
Assume $R$ is fixed, but large enough so that $B_R(0)$ contains $z^*_t(y) =: z^*$ and the compact set
$\{ z : \Lambda(y,z) \le \Lambda(y,z^*) +1 \}$.
We have (dropping the fixed index $t$ from the notation),
\[
\begin{aligned}
| \esp[\phi(Z^{\eps})|Y^{\eps}=y] - \phi(z^*) | &=
\left| \esp[\phi(Z^{\eps}) 1_{ \{ |Z^{\eps}| \le R \} } |Y^{\eps}=y] - \phi(z^*)
+ \esp[\phi(Z^{\eps}) 1_{ \{ |Z^{\eps}| > R \} } |Y^{\eps}=y] \right|
\\
&\le 
| \esp[\phi_R(Z^{\eps}) 1_{ \{ |Z^{\eps}| \le R \} } |Y^{\eps}=y] - \phi(z^*)|
+
C \int_{|z| \ge R} (1 + |z|^k) \frac{p^{\eps}(y,z)}{f^{\eps}(y)} dz
\end{aligned}
\]
By Theorem \ref{t:conditionalLaw}, the first term tends to $|\phi_R(z^*)-\phi(z^*)|=|\phi(z^*)-\phi(z^*)|=0$ as $\eps \to 0$.
The second term can be bounded as in Step 3 of the proof of Theorem \ref{t:conditionalLaw}, that is
\[
\begin{aligned}
\limsup_{\eps \to 0} \eps^2 \log \int_{|z| \ge R} (1 + |z|^k) \frac{p^{\eps}(y,z)}{f^{\eps}(y)} dz
&\le
- \liminf_{\eps \to 0} \eps^2 \log f^{\eps}(y)
+
\limsup_{\eps \to 0} \eps^2 \log \int_{ |z| \ge R} (1 + |z|^k) p^{\eps}(y,z) dz
\\
&\le
\Lambda(y,z^*) - \inf \{ \Lambda(y,z): |z| \ge R \}
\\
&\le \Lambda(y,z^*)-\Lambda(y,z^*)-1 = -1.
\end{aligned}
\]
where we have used Proposition \ref{p:tailInt} and estimate \eqref{e:marginVaradh} in the second step, and the choice of $R$
to conclude.
The last inequality implies that the left hand side vanishes as $\eps \to 0$.
\end{proof}
\medskip

Note that Corollary \ref{c:polGrowth} can be straightforwardly extended to the case of finitely many minimizers (as described in Remark \ref{r:finitelyMany}).

\section{Varadhan's formula for marginal densities} \label{s:varadh}

As a by-product of the estimates presented so far (Theorem \ref{t:BAL}, Proposition \ref{p:tailInt} and
Proposition \ref{p:marginVaradh}), it is possible to show that a Varadhan-type formula holds for the density of the projected diffusion $Y^{\eps}_t$.
Proposition \ref{p:marginVaradh} already provides a lower bound; we are left with the proof of the corresponding upper bound.
In the following theorem, the case $l=n$ recovers the classical Varadhan's formula \cite{VarB}, or rather L\'eandre's extension \cite{LeMaj,LeMin} to the hypoelliptic setting.

\begin{theorem} \label{t:VaradhMargin}
Consider $X^{\varepsilon} = \left( Y^{\varepsilon },Z^{\varepsilon }\right)$ the strong solution to
\[
dX^{\varepsilon}_t=b_{\varepsilon }\left( X^{\varepsilon}_t\right)
dt+\varepsilon \sum_{j=1}^d \sigma_j \left( X^{\varepsilon}_t\right) dW^j_t, \quad X_{0}^{\varepsilon
}=x_{0}^{\varepsilon},
\]
with $x_0^{\eps} \to x_0$, and assume strong H\"{o}rmander condition (sH) at all points.
Fix $y\in \mathbb{R}^{l}$ and $t>0$.
Then 
\begin{itemize}
\item[(i)] if $b_{\eps} \to 0$ as $\eps \downarrow 0$ in the sense of \eqref{e:convergDrift}, the density $f^{\eps}_t$ of $Y^{\eps}_t$ satisfies
\be \label{e:VaradhMarginSmallNoise}
\lim_{\epsilon \to 0} \eps^2 \log f^{\epsilon}_t (y) = -\Lambda_t(N_y).
\ee
where $\Lambda_t$ is defined by \eqref{e:rateFctX} with $b_0 \equiv 0$.

\item[(ii)] If $b_{\eps} \to b_0 \not \equiv 0$, assume that $C_{x_{0}}(h_0)$ is invertible for at least a minimizing control $h_0\in \mathcal{K}_{t}^{(y,z)}$, for every $z$ in a neighborhood of $z^{\ast}$, where $z^*$ is some (not necessarily strict, nor unique) global minimizer of the map $z \mapsto \Lambda_t(y,z)$.
Then, estimate \eqref{e:VaradhMarginSmallNoise} holds.
\end{itemize}
\end{theorem}
\medskip
\noindent
Consider a financial model $X=(Y,Z) \in \R \times \R^{n-1}$, where $Y$ models the log-price of an asset, and $Z$ its (possibly multi-dimensional) stochastic volatility.
The small-noise behavior of the logarithm of the density of $Y$ translates into the leading order asymptotic term of the implied volatility of European options (see for example Gatheral et al. \cite{GathLaurHsu} for an implementation of this approach in the small-maturity limit for 1D local volatility models).

\begin{remark} \label{r:marginal}
\emph{
Point (i) of Theorem \ref{t:VaradhMargin} covers the small-time case: setting $b_{\eps} = \eps^2 b$ and $x_0^{\eps} = x_0$, Brownian scaling yields $X^{\eps}_1 \sim X_{\eps^2}$ where $X_t =(Y_t, Z_t)$ is the solution to \eqref{e:SDEintro}, hence \eqref{e:VaradhMarginSmallNoise} is equivalent to the small-time Varadhan formula 
\[
\lim_{t \to 0} t \log f_t (y) = -\Lambda_1(N_y)
\]
with $f_t$ the pdf of $Y_t$.
}
\end{remark}

\textbf{Comparison with the marginal density expansions of Deuschel et al. \cite{DFJV:I}.}
A sufficient condition for \eqref{e:VaradhMarginSmallNoise} to hold is given as condition (ND) in Deuschel et al. \cite{DFJV:I}, see their Definition 2.7.
This condition - the finite dimensional analogue of the infinite dimensional condition (D) in Takanobu--Watanabe - appears as a generalized ``not in cut-locus'' condition from Riemannian geometry, and actually allows to derive a full expansion for the marginal density $f_t^{\eps}(y)$ as $\eps \to 0$, of the form
\begin{equation} \label{e:margExp}
f_t^{\eps}(y) = \eps^{-l} e^{-\Lambda(N_y)/\eps^2} e^{-\tilde{\Lambda}(y)/\eps} \bigl( c_0 + o(\eps) \bigr),
\end{equation}
expansion of which formula \eqref{e:VaradhMarginSmallNoise} captures only the leading-order exponential term $e^{-\Lambda(N_y)/\eps^2}$ (we refer to \cite[Theorem 2.8]{DFJV:I} for an account of the additional term $\tilde{\Lambda}(y)/\eps$).
In addition to the first order condition of invertibility of the deterministic Malliavin matrix along the minimizing controls in $\mathcal{K}_t^{(y,\cdot)}$, the condition (ND) in \cite{DFJV:I} requires to check a second-order condition corresponding to the non-degeneracy of the action along the minimizers.
The geometric interpretation is the non-focality of $x_0$ for the arrival subspace $N_y$.
Theorem \ref{t:VaradhMargin} precisely tells that this non degeneracy condition is not necessary in order to establish the asymptotic behaviour of the density on the log-scale as in \eqref{e:VaradhMarginSmallNoise}.
\bigskip

\begin{proof}[Proof of Theorem \ref{t:VaradhMargin}]
$(ii)$ Let us establish the upper bound
\[
\limsup_{\epsilon \to 0} \eps^2 \log f^{\epsilon}_t (y) \le -\Lambda_t(N_y).
\]
From the hypotheses in $(ii)$ and Theorem \ref{t:BAL}, we know that there exists $\eta>0$ such that estimate \eqref{e:BALup} holds uniformly on $B_{\eta}(z^*)$.
Write (omitting the fixed index $t$)
\begin{equation} \label{e:upDecomp}
\int_{\R^{n-l}} p^{\eps}(y,z) dz = \int_{B_{\eta}(z^*)} p^{\eps}(y,z) dz + \int_{B_{\eta}^c(z^*)} p^{\eps}(y,z) dz.
\end{equation}
It follows from estimate \eqref{e:BALup} that for every $\delta>0$ we can find
$\eps_0 = \eps_0(\delta,\eta)$ such that
\begin{equation} \label{e:encadrementUp}
p^{\eps}(y,z)
\le \exp \Bigl( \frac{-\Lambda(y,z) + \delta}{\eps^2} \Bigr)
\end{equation}
for all $z \in B_\eta(z^*)$ and $\eps < \eps_0$.
For such values of $\eps$, one has
\[
\begin{aligned}
\int_{B_\eta(z^*)} p^{\eps}(y,z) dz \le
e^{ \frac{\delta}{\eps^2} } \int_{ B_\eta(z^*) } e^{ -\frac{ \Lambda(y,z) } {\eps^2} } dz
\le e^{ \frac{\delta}{\eps^2} } e^{-\frac{ \Lambda(y,z^{\ast}) } {\eps^2} } \lambda^{n-l}( B_\eta(z^*) )
\end{aligned}
\]
where the last inequality trivially follows from $\Lambda(y,z) \ge \Lambda(y,z^{\ast})$.
Therefore,
\be \label{e:upCentral}
\limsup_{\eps \to 0} \eps^2 \log \int_{B_\eta(z^*)} p^{\eps}(y,z) dz \le - \Lambda(y,z^{\ast}) + \delta.
\ee
Now, using Proposition \ref{p:tailInt}, one has
\be \label{e:upTail}
\limsup_{\eps \to 0} \eps^2 \log \int_{B^c_\eta(z^*)} p^{\eps}(y,z) dz \le
-\inf \{ \Lambda(y,z'): z' \in B^c_\eta(z^*) \}
\le -\Lambda(y,z^{\ast}).
\ee
Therefore, taking $\limsup_{\eps \to 0} \eps^2 \log$ in (\ref{e:upDecomp}) and using estimates \eqref{e:upCentral} and \eqref{e:upTail}, one gets
\[
\limsup_{\eps \to 0} \eps^2 \log f^{\eps}(y) =
\limsup_{\eps \to 0} \eps^2 \log \int_{\R^{n-l}} p^{\eps}(y,z) dz
\le
\max(-\Lambda(y,z^{\ast})+\delta, -\Lambda(y,z^{\ast}))
= -\Lambda(y,z^{\ast})+\delta.
\]
Since $\delta$ was arbitrary, the right hand side can be improved to $-\Lambda(y,z^{\ast}) = -\Lambda(N_y)$, as claimed.
\\
The lower bound $\liminf_{\eps \to 0} \eps^2 \log f_t^{\eps}(y) \ge -\Lambda_t(N_y)$ was already obtained in Proposition \ref{p:marginVaradh} (see Remark \ref{r:lowerBoundMinimizerNotUnique} for the case where the minimizer $z^*$ is not
unique).

$(i)$ Under strong \Horm condition, together with $b_0 = \lim_{\eps \to 0}b_{\eps} \equiv 0$, the function $\Lambda_t$ is continuous on $\R^n$ (see Remark \ref{r:noDriftLambdaContinuous}), and the two functions $\Lambda_t$ and $\Lambda_{R,t}$ coincide everywhere (see Comment \ref{c:BALresult}$(ii)$).
On the one hand, thanks to the continuity of $\Lambda_t$, estimate \eqref{e:BALup} in Theorem \ref{t:BAL} holds uniformly over $x$ in compact sets of $R^n$ (and not only locally around $(y,z^*)$).
In the proof above, the upper bounds $\limsup_{\epsilon \to 0} \eps^2 \log f^{\epsilon}_t (y) \le -\Lambda_t(N_y)$ only relies on estimate \eqref{e:BALup} and Proposition \ref{p:tailInt}, and can be proven exactly as done above.
On the other hand, identity \eqref{e:actionsCoincide} in Proposition \ref{p:marginVaradh} holds, and the proof of Proposition \ref{p:marginVaradh} can be rerun with no modifications, leading to the lower bound $\liminf_{\eps \to 0} \eps^2 \log f_t^{\eps}(y) \ge -\Lambda_t(N_y)$, which proves \eqref{e:VaradhMarginSmallNoise}.
\end{proof}

\section{Applications: asymptotics of local volatilities} \label{s:localVol}

In this section we focus on the stochastic volatility model
\begin{equation} \label{e:stochVol}
\begin{aligned}
dY_{t} &= -\frac12 Z_{t}^2 dt + Z_t dB^{1}_{t}, &Y_{0}=0,
\\ \
dZ_{t} &= \beta(Z_{t})dt + \alpha(Z_t) dB^{2}_{t}, &Z_0 \neq 0
\end{aligned}
\end{equation}
where the process $Y$ models the log-value of a financial asset, and $Z$ its stochastic volatility.
Setting $S_t=S_0 e^{Y_t}$ leads to the familiar ``Black-Scholes with stochastic volatility'' dynamics $dS_t = S_t Z_t dB^1_t$.
Here $B_t=(B^1_t, B^2_t)$ is a two-dimensional Brownian motion with correlated components, in short $d \langle B^1, B^2 \rangle_t = \rho dt$ for some $\rho \in (-1,1)$.
$B$ can be obtained from a two-dimensional standard Brownian motion $W$ by setting $B = \sqrt{\Gamma} W$, where $\sqrt{\Gamma}$ is a (any) choice of the square root\footnote{A typical choice in this setting is provided by the Cholesky decomposition $\sqrt{\Gamma} \sqrt{\Gamma}^* =\Gamma$, with $\sqrt{\Gamma}=\left(\begin{array}{c c} \rho & \sqrt{1-\rho^2} \\ 1 & 0 \end{array}\right)$.} of the correlation matrix $\Gamma=\left(\begin{array}{c c} 1 & \rho \\ \rho & 1 \end{array}\right)$.
\\
Let us note straight away that the diffusion vector fields in \eqref{e:stochVol} read
\be \label{e:sigmaStochVol}
\sigma_1(z) = \left(\begin{array}{c} z \sqrt{\Gamma}_{11} \\ \alpha(z) \sqrt{\Gamma}_{21} \end{array} \right); \qquad
\sigma_2(z) = \left(\begin{array}{c} z \sqrt{\Gamma}_{12} \\ \alpha(z) \sqrt{\Gamma}_{22} \end{array} \right).
\ee
While it is clear that the couple $\sigma_1(z),\sigma_2(z)$ spans $\R^2$ at every $z \neq 0$ such that $\alpha(z) \neq 0$ (recall $\sqrt{\Gamma}$ is invertible under the assumption $\rho \neq \pm 1$), this condition fails on the set $\{z=0\}$, whatever value the function $\alpha$ takes there.
The model \eqref{e:stochVol} then naturally fits into the non-elliptic framework.

\begin{example} \label{e:stochVolModels}
\emph{
A relevant parametric choice of the drift term in \eqref{e:stochVol} is given by the affine function $\beta(z)=a+bz$, which has the typical mean-reverting form $\beta(z)= |b|(a/|b| - z)$ when $b<0$.
}
\end{example}

\subsection{Extension of main results to Stochastic Volatility models} \label{s:stocVol}

Postponing for a moment the precise assumptions on the coefficients $\alpha,\beta$, let us first describe the different types of asymptotic problems that can arise in the applications.
The following class of small-noise equations embeds both small-time and (in some cases) space-asymptotic problems:
\begin{equation} \label{e:smallNoiseStochVol}
\begin{aligned}
dY^{\eps}_t &= -\frac12 \eps^{\theta} \left(Z^{\eps}_t \right)^2 dt + \eps Z^{\eps}_t dB^1_t, &Y^{\eps}_{0}=0,
\\ 
dZ^{\eps}_t &= \beta_{\eps} \left(Z^{\eps}_t\right)dt + \eps \alpha\left(Z^{\eps}_t\right) dB^2_t, &Z^{\eps}_0 = z^{\eps}_0.
\end{aligned}
\end{equation}
Here, $\theta \ge0$ is a parameter that depends on the asymptotic regime under consideration (we will have $\theta \in \{0,2\}$ in our applications),
\[
z^{\eps}_0 \to z_0 \qquad \mbox{as } \eps \to 0
\]
and, analogously to \eqref{e:convergDrift},
\be \label{e:driftConvergenceStochVol}
\beta_{\eps} \to \beta_0 \qquad \mbox{as } \eps \to 0
\ee
for some limiting function $\beta_0$, in the sense of uniform convergence on compact sets of $\R$ together with the derivatives of any order.
We also assume that the sequence of norms $|\partial^k_{z} \beta_{\eps}|_{\infty}$ is uniformly bounded in $\eps$, for every $k \ge 0$.
The associated limiting controlled system reads
\be \label{e:ODEstochVol}
\begin{aligned}
d\varphi^{(1)}_t &= -\frac12 1_{\theta=0} \bigl(\varphi^{(2)}_t\bigr)^2 dt + \varphi^{(2)}_t d(\sqrt{\Gamma} h)^{(1)}_t, \quad \varphi^{(1)}_0 = 0,
\\
d\varphi^{(2)}_t &= \beta_0 \bigl(\varphi^{(2)}_t \bigr)dt + \alpha\bigl(\varphi^{(2)}_t\bigr) d(\sqrt{\Gamma} h)^{(2)}_t, \quad \varphi^{(2)}_0 = z_0;
\end{aligned}
\ee
where $h = (h^1,h^2)$ is a two-dimensional control and $\Gamma$ the correlation matrix in \eqref{e:stochVol}.
Let us denote
\be \label{e:rateFctXstochVol}
\Lambda_t^{SV}(y,z)= \inf\left\{\frac12 |h|_H^2: h \in \mathcal{K}_t^{(y,z)} \right\}, \qquad (y,z) \in \R^2
\ee
the action of the system \eqref{e:ODEstochVol}, where $\mathcal{K}_t^{(y,z)} =\{ h : (\varphi^{(1)}_0,\varphi^{(2)}_0)=(0,z_0), (\varphi^{(1)}_t,\varphi^{(2)}_t)=(y,z) \}$.
\medskip

We assume that the coefficients $\beta_{\eps},\beta_0$ and $\alpha$ satisfy:
\begin{itemize}
\item[(SV)] $\beta_{\eps},\beta_0, \alpha: \R \to \R$ are smooth and Lipschitz functions, with $\alpha(0) \neq 0$.
\end{itemize}
The application of Theorem \ref{t:conditionalLaw} to the system \eqref{e:smallNoiseStochVol} is a priori not justified, because of the lack of global boundedness for the coefficients of the SDE (and their derivatives).
In this respect, let us note that, even if a boundedness assumption were in force for $\alpha$ and $\beta_{\eps}$, the two-dimensional system \eqref{e:smallNoiseStochVol} would still \emph{not} have bounded coefficients (because of the terms $z$ and $-\frac12 z^2$ in the equation for the $Y$ component ).
Nevertheless, one can exploit the Lipschitz condition in (SV) (which is rather mild in this setting) in order to extend our main result on the asymptotics of conditional expectations.

\begin{theorem}[\textbf{Small noise asymptotics of local volatility: the general case}] \label{t:conditionalLawStochVol}
Assume condition (SV) on the coefficients $\beta_{\eps}, \beta_0$ and $\alpha$, and denote $(Y^{\eps}, Z^{\eps})$ the unique strong solution to \eqref{e:smallNoiseStochVol} with $\theta \ge 0$.
Fix $t>0$ and $y \in \R$, and assume there exist a unique $z^*=z^*_t(y)$ minimizing the action $\Lambda_t^{SV}$ in \eqref{e:rateFctXstochVol} on the set $N_y=(y,\cdot)$, and that there are finitely many minimizing controls in the set $\mathcal{K}_t^{(y,z^*)}$.
%that is
%\[
%\mathcal{K}^{(y,z^*)}_{\min} = \Bigl\{ h_0 \in \mathcal{K}^{(y,z^*)} : \frac12|h_0^2|_H = \Lambda_t^{SV}(N_y) \Bigr\} 
%\]
%is finite.
If one of the following conditions is satisfied:
\begin{itemize}
\item[(i)] $z_0>0$ and $\alpha(z_0)\neq 0$;
\item[(ii)] $z_0=0$ and $y \neq 0$;
\end{itemize}
then, for all functions $\phi \in C( \mathbb{R}^{n-l})$ with polynomial growth,
\be \label{e:stochVolsmallNoisePhi}
\esp[\phi(Z^{\eps}_t) | Y^{\eps}_t = y] \to \phi\left(z^*_t(y)\right) \qquad \mbox{as } \eps \downarrow 0.
\ee
In particular,
\be \label{e:stochVolsmallNoiseLocVol}
\sigma_{loc}^{\eps}(t,y)^2 = 
\esp[\left(Z^{\eps}_t\right)^2 | Y^{\eps}_t = y] \to \left(z^*_t(y)\right)^2 \qquad \mbox{as } \eps \downarrow 0.
\ee
Finally, if there exist finitely many minimizers $z^{*}_i$ for $\Lambda^{SV}_t$ on $N_y$ (each one associated with finitely many minimizing controls $h_0 \in \mathcal{K}_t^{(y,z^{*}_i)}$), and one of conditions $(i)$ or $(ii)$ is satisfied, then the limits in \eqref{e:stochVolsmallNoisePhi} and \eqref{e:stochVolsmallNoiseLocVol} are replaced with $\sum_{i=1}^N a_i \phi(z^{*}_i)$, respectively $\sum_{i=1}^N a_i (z^{*}_i)^2$, for some weights $a_i \ge 0$ such that $\sum_{i=1}^N a_i = 1$.
\end{theorem}

\begin{remark}
\emph{It will be seen in the next two sections that conditions $(i)$ or $(ii)$ in the previous theorem are met in specific applications to small-time and space asymptotics of local volatilities, see Theorems \ref{t:smallTimeLocalVol} and \ref{r:steinSteinWings}.}
\end{remark}

\begin{proof}
Let $b_{\eps}^R$, $b_0^R$, and $\alpha^R$, $R>0$, be smooth and bounded extensions of $b_{\eps}|_{B_R(0)}$, $b_0|_{B_R(0)}$, and $\alpha|_{B_R(0)}$ respectively (see the Appendix \ref{s:app0} for a precise definition of such extensions), and denote $(Y^{\eps,R}, Z^{\eps,R})$ the unique strong solution to \eqref{e:smallNoiseStochVol} when $b_{\eps}$ and $\alpha$ are replaced with $b_{\eps}^R$ and $\alpha^R$.
Also denote $\Lambda_t^R$ the action function associated to the limiting deterministic system \eqref{e:ODEstochVol} with coefficients $b_0^R, \alpha^R$.
Set
\[
\tau_R^{\eps} = \inf\{s \ge 0 : |Z^{\eps}_s| > R \}.
\]
On the event $\{ \tau_R^{\eps} > t\}$, by the pathwise uniqueness for the second equation in \eqref{e:smallNoiseStochVol}, $(Z^{\eps}_s)_{s \le t}$ is indistiguishable from $(Z^{\eps,R}_s)_{s \le t}$.
In addition,
\[
Y^{\eps}_t = -\frac12 \eps^{\theta} \int_0^t \left(Z^{\eps}_s \right)^2 ds + \eps \int_0^t Z^{\eps}_s dB^1_s
= -\frac12 \eps^{\theta} \int_0^t \bigl(Z^{\eps,R}_s \bigr)^2 ds + \eps \int_0^t Z^{\eps,R}_s dB^1_s
=Y^{\eps,R}_t, \quad a.s.
\]
Note that the assumption $\alpha(0)\neq 0$ guarantees that strong \Horm condition is satisfied at all points for the SDE \eqref{e:smallNoiseStochVol}: indeed, simple calculations show that
\[
[\sigma_1, \sigma_2 ](z) = -\det(\sqrt{\Gamma}) \left(\begin{array}{c} \alpha(z) \\ 0 \end{array} \right),
\]
therefore the vectors
\[
\sigma_1(0) = \left(\begin{array}{c} 0 \\ \alpha(0) \sqrt{\Gamma}_{21} \end{array} \right); \qquad
\sigma_2(0) = \left(\begin{array}{c} 0 \\ \alpha(0) \sqrt{\Gamma}_{22} \end{array} \right); \qquad
[\sigma_1, \sigma_2 ](0) = -\det(\sqrt{\Gamma}) \left(\begin{array}{c} \alpha(0) \\ 0 \end{array} \right)
\]
span the full $\R^2$ (when the assumption $\det(\Gamma) = 1-\rho^2 \neq 0$ is in force).
\\
Denoting $|\phi|_{\infty,R} := \sup_{|z| \le R} |\phi(z)|$, one has
\be \label{e:decomposition}
\begin{aligned}
\left| \esp[ \phi(Z^{\eps}_t)|Y^{\eps}_t=y] - \phi(z^*) \right| &\le
\esp[ |\phi(Z^{\eps}_t)- \phi(z^*)| 1_{ \tau^{\eps}_R > t} | Y^{\eps}_t=y]
\\
&\quad+ 
\esp[ |\phi(Z^{\eps}_t)- \phi(z^*)| 1_{ \tau^{\eps}_R \le t} 1_{|Z^{\eps}_t|\le R} | Y^{\eps}_t=y]
\\
&\quad+ 
\esp[ |\phi(Z^{\eps}_t)- \phi(z^*)| 1_{ \tau^{\eps}_R \le t } 1_{|Z^{\eps}_t|>R} | Y^{\eps}_t=y]
\\
&\le \esp[ |\phi(Z^{\eps,R}_t)- \phi(z^*)| 1_{ \tau^{\eps}_R > t } | Y^{\eps,R}_t=y]
+
2 |\phi|_{\infty,R} \: \Prob\left( \tau^{\eps}_R \le t | Y^{\eps}_t=y \right)
\\
&\quad+ C \int_{|z|>R} (1+|z|^k) \frac{p^{\eps}(y,z)}{f^{\eps}(y)} dz
\\
&\le \esp[ (\phi(Z^{\eps,R}_t)- \phi(z^*))^2 | Y^{\eps,R}_t=y]^{1/2}
+
2 |\phi|_{\infty,R} \: \Prob\left( \sup_{s \le t} |Z^{\eps}_s| \ge R | Y^{\eps}_t=y \right)
\\
&\quad+ C \int_{|z|>R} (1+|z|^k) \frac{p^{\eps}(y,z)}{f^{\eps}(y)} dz
\\
&=: \epsilon_1(\eps,R) + \epsilon_2(\eps,R) + \epsilon_3(\eps,R).
\end{aligned}
\ee
where we have used \Hold's inequality in the last step.
We will now show that taking $R$ large enough, but fixed, the right hand side tends to zero as $\eps \to 0$.

$\epsilon_1$: Lemma \ref{l:rateFctLoc} in Appendix \ref{s:app0} establishes that $R$ can be chosen such that $z^*$ is also the unique minimum point of the function $z \mapsto \Lambda^R(y,z)$.
Assume that Condition $(i)$ is satisfied. The vector fields $\sigma_1(z), \sigma_2(z)$ defined in \eqref{e:sigmaStochVol} span the whole $\R^2$ at the starting point $z_0$: this is enough to establish, see Lemma \ref{l:invertibilityLocalEllipt}, that the covariance matrix $C_{(0,z_0)}(h)$ is invertible for \emph{all} $h$.
Otherwise, assume that Condition $(ii)$ is satisfied.
From the continuity of $\alpha$ and the condition $\alpha(0) \neq 0$ in (SV), $\alpha$ is bounded away from zero in a neighbourhood $V$ of $z_0=0$.
A simple inspection of the limiting (truncated) controlled system
\[
\begin{aligned}
d\varphi^{(1,R)}_t &= -\frac12 \bigr(\varphi^{(2,R)}_t\bigr)^2 dt + \varphi^{(2,R)}_t d(\sqrt{\Gamma} h)^1_t, \quad \varphi^{(1,R)}_0 = 0
\\
d\varphi^{(2,R)}_t &= \beta_0^R\bigr(\varphi^{(2,R)}_t\bigr) dt + \alpha^R\bigr(\varphi^{(2,R)}_t\bigr) d(\sqrt{\Gamma} h)^2_t, \quad \varphi^{(2,R)}_0 = z_0
\end{aligned}
\]
shows that $\varphi^{(2,R)}_s=0$ for all $s \in [0,t]$ implies $y = \varphi^{(1,R)}_t = 0$.
Therefore, the second coordinate of the controlled path $z_s=\varphi^{(2,R)}_s$ must cross the set $V \setminus \{0\}$, where the diffusion vector fields are elliptic, in order to have $y_t = y \neq 0$.
Lemma \ref{l:invertibilityLocalEllipt} then allows to conclude that $C_{(0,z_0)}(h)$ is invertible for every $h \in \mathcal{K}_t^{(y,\cdot)}$.
In both cases, the hypotheses of Theorem \ref{t:conditionalLaw} are satisfied, yielding $\epsilon_1(\eps,R) \to |\phi(z^*)-\phi(z^*)|= 0$ as $\eps \to 0$.

$\epsilon_2$: Any optimal control $h_0 \in \mathcal{K}_t^{(y,z^*)}$ satisfies $\frac12 |h_0|_H^2 = \Lambda_t(y,z^*) = \Lambda_t(N_y)$.
Estimate \eqref{e:gronwallSolutionMap} in Lemma \ref{l:propertiesLambda} then implies
\[
\sup_{s \le t} |\varphi^{(2)}_s(h_0)| \le C e^{C |h_0|_H} = C e^{C \sqrt{2 \Lambda_t(N_y)}} := \hat{R}
\]
where $C=C(t,z_0,K)$ is the constant defined in Lemma \ref{l:propertiesLambda}, and $K$ is a common Lipschitz constant for $\beta_0$ and $\alpha$.
It follows from Theorem \ref{t:conditionalLaw} and the subsequent Remark \ref{r:finiteDim} that the law of $(Z^{\eps}_s, s \le t)$ conditional on $Y^{\eps}_t=y$ converges weakly to a law supported by the finitely many paths $\{ \varphi^{(2)}_{\cdot}(h_0): h_0 \in \mathcal{K}^{(y,z^*)}_{min} \}$. 
Taking $R > \hat{R}+1$, it is clear that
\[
\Prob\left( \sup_{s \le t} |Z^{\eps}_s| \ge R \big| Y^{\eps}_t=y \right) \to 0, \quad \mbox{as } \eps \to 0,
\]
therefore $\epsilon_2(\eps,R) \to 0$ as well.

$\epsilon_3$: The integral term in $\epsilon_3(\eps,R)$ also appears in the proof of Corollary \ref{c:polGrowth}, and can be bounded exactly as done there.
\\
The proof in the case of finitely many minimizers $z^{*}_i$ goes through the same steps, using in \eqref{e:decomposition} the fact that $\{z^*_i\}_i$ is also the set of global minimizers of $z \mapsto \Lambda^R_t(y,z)$ when $R$ is large enough (see Lemma \ref{l:rateFctLoc} in the Appendix).
\end{proof}

\subsection{Small-time behavior and Berestycki, Busca and Florent \cite{BBFstochVol2004} asymptotics of efficient volatility revisited}
\label{s:shorTimeLocalVol}

The short time behavior of the local volatility function obtained as a projection of stochastic volatility was addressed by Berestycki, Busca and Florent \cite[section 5]{BBFstochVol2004}, who use local volatility as an intermediate step in the computation of the implied volatility of European options.
Using Theorem \ref{t:conditionalLawStochVol}, one can give a generalization of their result, stated for stochastic volatility models with bounded and uniformly elliptic coefficients, to hypoelliptic models with unbounded coefficients.

\begin{theorem} \label{t:smallTimeLocalVol}
Assume that the coefficients $\beta, \alpha$ in \eqref{e:stochVol} are smooth and Lipschitz functions with $\alpha(0)\neq 0$ and $\alpha(Z_0)\neq 0$, and consider the unique strong solution $(Y,Z)$ to \eqref{e:stochVol}.
Let $\Lambda_t^{SV}(y,z)$ be the action function of the system \eqref{e:ODEstochVol} when $b_0 \equiv 0$, $\theta = 2$ and $z_0=Z_0$, that is
\[
\Lambda_t^{SV}(y,z)= \inf \Bigl\{\frac12 |h|_H^2: (\varphi^{(1)}_0,\varphi^{(2)}_0)=(0,Z_0); \: (\varphi^{(1)}_t,\varphi^{(2)}_t)=(y,z) \Bigr\}
\]
\[
\begin{aligned}
d\varphi^{(1)}_t &= \varphi^{(2)}_t d\Bigl(\sqrt{\Gamma} h\Bigr)^1_t, \quad \varphi^{(1)}_0 = 0,
\\
d\varphi^{(2)}_t &= \alpha(\varphi^{(2)}_t) d\Bigl(\sqrt{\Gamma} h\Bigr)^2_t, \quad \varphi^{(2)}_0 = Z_0.
\end{aligned}
\]
Fix $y \in \R$ and assume there exists a unique $z^*=z^*(y)$ minimizing the function $\Lambda_1^{SV}$ on the line $N_y=(y,\cdot)$, and that there are finitely many minimizing controls in $\mathcal{K}_t^{(y,z^*)}$.
Then
\be \label{e:shortTimeLocalVol}
\sigma^2_{loc}(t,y) = 
\esp[\left(Z_t\right)^2 |Y_t=y] \to z^*(y)^2 \qquad \mbox{as } t \to 0.
\ee
In the presence of finitely many minimizers $z^*_i(y)$, $i=1,\dots,N$ (each one associated with finitely many minimizing controls in $\mathcal{K}_t^{(y,z^*_i)}$), convergence holds towards a convex combination of the $z^*_i(y)^2$:
\be \label{e:shortTimeLocalVolFinitelyMany}
\sigma^2_{loc}(t,y) 
\to \sum_{i=1}^N a_i z^*_i(y)^2 \qquad \mbox{as } t \to 0.
\ee
\end{theorem}

\begin{proof}
For every $\eps > 0$, the process $(Y^{\eps}_t, Z^{\eps}_t)_{t \ge 0}:= (Y_{t \eps^2}, Z_{t \eps^2})_{t \ge 0}$ has the same law as the solution of the SDE \eqref{e:smallNoiseStochVol} with 
\[
\beta_{\eps}(z)=\eps^2 \beta(x); \qquad \theta=2; \qquad z^{\eps}_0=Z_0, \quad \forall \eps>0.
\]
The functions $\alpha$, $\beta_{\eps}$ and $\beta_0=\lim_{\eps \to 0} \beta_{\eps}\equiv 0$ clearly satisfy assumption (SV).
Therefore, the hypotheses of Theorem \ref{t:conditionalLawStochVol}, case $(i)$, are satisfied, and Theorem \ref{t:smallTimeLocalVol} follows.
\end{proof}
\medskip

Once \eqref{e:shortTimeLocalVol} is at hand, proceeding along the lines of \cite{BBFstochVol2004}, one could retrieve the well-known short-time limit for the implied volatility, namely (see \cite[Theorem 1.2]{BBFstochVol2004})
\be \label{e:BBFimpliedVol} 
\sigma_{BS}(t,y) \to \frac{|y|}{\sqrt{2\Lambda^{SV}_1(N_y)}}
\qquad \mbox{as } t \downarrow 0,
\ee
where $\sigma_{BS}(t,y)$ is the Black-Scholes implied volatility of a Call option with maturity $t$ and log-moneyness $y=\log(K/S_0)$.
In this sense, Theorem \ref{t:smallTimeLocalVol} does not yield a refinement of the asymptotics of implied volatility in \cite{BBFstochVol2004}, but rather allows to weaken the hypotheses on the model coefficients.
As addressed in Section \ref{s:varadh}, a more direct way to prove \eqref{e:BBFimpliedVol} would be to use the short-time version of Varadhan's formula \ref{e:VaradhMarginSmallNoise} for the marginal density of the projection $Y_t$, see Remark \ref{r:marginal}.
\bigskip

\textbf{Comments on heat kernel expansions and the Laplace method.}
We compare our approach to Theorem \ref{t:smallTimeLocalVol} to the route followed in Henry-Labord\`ere \cite[Chap. 6]{PH2}.
Starting from a small-time heat kernel expansion
\be \label{e:smallTimeHKE}
p_t(x) = \frac1{2\pi t} e^{-\frac{\Lambda(x)}{t}}(c_0(x) + O(t;x)) \qquad \mbox{as } t \downarrow 0,
\ee
and assuming that for some $y$, the map $z \mapsto \Lambda(y,z)$ has a unique minimizer $z_y^*$ such that $\partial_{zz} \Lambda(y,z^*_y) > 0$, an heuristic application of the Laplace method yields:
\be \label{e:Laplace}
\begin{aligned}
\sigma_{\mathrm{loc}}(t,y)^2 = \esp\left[Z_t^2 | Y_t = y \right]
&= \frac{\int z^2 \ p_t(y,z) dz}{\int p_t(y,z) dz} 
\\ 
&= \frac{\int z^2 e^{-\frac{\Lambda(y,z)}{t}}(c_0(y,z) + O(t;x)) dz}
{\int e^{-\frac{\Lambda(y,z)}{t}}(c_0(y,z) + O(t;x)) dz}
\\
&\sim \frac{ (z_y^*)^2 \ e^{-\frac{\Lambda(y,z^*)}t} c_0(y,z^*_y) \: \sqrt{2 \pi t} \bigl( \partial_{zz} \Lambda(y,z^*) \bigr)^{-1/2} }
{e^{-\frac{\Lambda(y,z^*)}{t}} c_0(y,z^*_y) \: \sqrt{2 \pi t} \bigl( \partial_{zz} \Lambda(y,z^*) \bigr)^{-1/2} },
\qquad \mbox{as } t \downarrow 0
\\
&= (z_y^*)^2
\end{aligned}
\ee
in agreement with \eqref{e:shortTimeLocalVol}.
Of course, here we have plugged the expansion \eqref{e:smallTimeHKE}, which is know to hold uniformly on compact sets (that do not intersect the cut-locus), and then integrated on the whole line, neglecting the tail contributions to the integrals in \eqref{e:Laplace}.
The condition $\partial_{zz} \Lambda(y,z^*) > 0$, typical from Laplace asymptotics, plays the role of the second-order `non-focality' condition in Deuschel et al. \cite[Definition 2.7]{DFJV:I}.
As pointed out in our discussion after Remark \ref{r:marginal}, we do not rely here on such a non-degeneracy assumption. 
Moreover, the main message of Theorems \ref{t:conditionalLawStochVol} and \ref{t:smallTimeLocalVol} is that \emph{the asymptotic behaviour of the logarithm of the density} is enough to establish the leading order term of the local volatility function.
On the other hand, when a full heat kernel expansion is available as in \eqref{e:smallTimeHKE}, the Laplace method allows to provide higher-order terms in \eqref{e:Laplace}; see \cite[Chap. 6]{PH2} (where an ellipticity assumption is considered).

\subsection{Asymptotic slopes of local volatility in the Stein-Stein model} \label{s:asympSlopes}

In the Stein--Stein model \cite{SteinStein} (see Sch{\"o}bel and Zhu \cite{SchZhu} for the correlated case $\rho\neq 0$), the volatility process follows an Ornstein-Uhlenbeck process:
\be \label{e:steinStein}
\begin{aligned} 
dY_{t} &= -\frac12 Z_{t}^2 dt + Z_t dB^{1}_{t}, \qquad Y_{0}=0,
\\
dZ_{t} &= (a + bZ_{t})dt + c dB^{2}_{t}, \qquad Z_0=z_0>0,
\end{aligned}
\ee
with $a, b \in \R$, $c>0$ and $d \langle B^{1},B^{2} \rangle_t=\rho dt$ with $\rho \in (-1,1)$.
The typical mean-reverting form of the drift coefficient is obtained when $a\geq 0$ and $b<0$.
In the following, we consider $b<0$ and $\rho \le 0$ (the typical configuration in Equity markets) in order to streamline the computations, but this restriction is not essential.
\\
Setting $Y^{\eps}_t := \eps^2 Y_t, Z^{\eps}_t := \eps Z_t$, the rescaled variables are seen to satisfy the small-noise problem
\be \label{e:steinSteinSmallNoise}
\begin{aligned} 
dY^{\eps}_{t} &= -\frac12 (Z^{\eps}_{t})^2 dt + \eps Z^{\eps}_{t} dW^{1}_{t}, \qquad Y^{\eps}_{0}=0,
\\
dZ^{\eps}_{t} &= (a\eps + b Z^{\eps}_{t}) dt + \eps \, c \, dW^{2}_{t}, \qquad Z^{\eps}_0=\eps z_0.
\end{aligned}
\ee
which belongs to the general class \eqref{e:smallNoiseStochVol} with $\theta=0$, $\beta_{\eps}(z)=a \eps + b z$, $\alpha(z)=c$.
Note that $z^{\eps}_0:=\eps z_0 \to 0$ as $\eps \to 0$, that is, we are in a situation where the limiting starting point $x_0=(y_0,z_0)=(0,0)$ belongs to the sub-elliptic set $\{z=0\}$.
\medskip

The Hamiltonian system associated to the Stein--Stein model was solved in \cite{DFJV:II}.
For every $y \neq 0$, the solution of the ODEs \eqref{e:hamiltonODE} subject to the boundary conditions
\[
\begin{aligned}
x_0=(0,0); \qquad &x_t=(y,\cdot), \ \ y \neq 0
\\
 & p_t=(\cdot,0)
\end{aligned}
\]
reveal the existence of two minimizing controls $h_0^{\pm} \in \mathcal{K}_t^{(y,\cdot)}$, yielding the two (symmetric) arrival points $\varphi^{h_0^{\pm}}_t = x^{\pm}_t = (y,z_t^{\pm}(y)) \in N_y$. 
Full details about explicit computations are found in \cite[Section 5.2]{DFJV:II}; the two minimizers $z_t^{\pm}(y)$ are given by
\be \label{e:optimalZ}
z_t^{\pm}(y) = \pm \frac{q(y) c^2 t}{r_1} \sin(r_1)
\ee
where
\be \label{e:q}
q(y) = \frac2c \left[ \frac{2r_1^3 y }{t^3 \left( \bigl(c^2(2p-1)-2\rho c \tilde{b}\bigr)(2r_1-\sin(2r_1))
+2\rho c r_1 (1-\cos(2r_1))/t \right) } \right]^{1/2}
\ee
with $\tilde{b}=b+\rho c p$ and $p = p(\overline{r},y)$, where
\be \label{e:p}
p(r,y) = \frac1{2(1-\rho^2)} \left[ \Bigl(1+2\rho \frac bc\Bigr) + \text{sign}(y) \sqrt{\Bigl(1+2\rho \frac bc\Bigr)^2
+ 4(1-\rho^2)\Bigl[\frac{b^2}{c^2} + \frac{r^2}{c^2 t^2}\Bigr] } \right],
\ee
and $\overline{r}=\overline{r}(y)$ is the smallest strictly positive solution to
\be \label{e:rRoot}
r \cos r - t (b + \rho c p(r,y) ) \sin r = 0.
\ee
The equation \eqref{e:rRoot} appears when imposing the transversality condition $p_t =(\cdot,0)$ from \eqref{e:transvers}.

\begin{remark}
\emph{
It is not difficult to show that $\overline{r}$ is the unique solution of equation \eqref{e:rRoot} in a bounded interval $I$, which is independent of the model parameters.
In practice, $\overline{r}$ can be found using some simple root-finding procedure (such as bisection or Newton method).
}
\end{remark}

\noindent
Applying Theorem \ref{t:conditionalLawStochVol} and the scaling leading to \eqref{e:steinSteinSmallNoise}, we are able to show that the local variance $\sigma^2_{\mathrm{loc}}(t,y)$ in the Stein-Stein model is asymptotically linear for large values of $|y|$, in a similar spirit to Lee's moment formula \cite{Lee} for the implied volatility (see also the subsequent refinements in \cite{GulForm}).

\begin{theorem}[\textbf{Local volatility `wings' in the Stein--Stein model}] \label{r:steinSteinWings}
Denote $z_t(y)$ the common absolute value of $z_t^{\pm}(y)$ in \eqref{e:optimalZ}.
The local volatility in the correlated Stein--Stein model \eqref{e:steinStein} satisfies, for any $t>0$,
\be \label{e:steinSteinWings}
\lim_{y \to \pm \infty} \frac{\sigma^2_{\mathrm{loc}}(t,y)}{|y|}
= 
\lim_{y \to \pm \infty} \frac1{|y|}
\esp[\left(Z_t\right)^2 |Y_t=y]
=
\left( z_t(\pm1) \right)^2
=
\Bigl( \frac{q c^2 t}{\overline{r}} \sin(\overline{r})\Bigr)^2
\ee
with $q=q(\pm 1)$ and $\overline{r}=\overline{r}(\pm1)$ according to the sign of $y$ in \eqref{e:steinSteinWings}, where $\overline{r}(y)$ is the smallest strictly positive solution of equation \eqref{e:rRoot} and the function $q$ is given in \eqref{e:q} above.
\end{theorem}
\medskip
\noindent
Note that the value of the limit in \eqref{e:steinSteinWings} does not depend on the initial volatility $z_0$, nor on the parameter $a$.

\begin{comment} \label{c:commentLocalVolWings}
\emph{The asymptotic formula \eqref{e:steinSteinWings} can be used to patch the numerical evaluation of the local volatility from Dupire's formula \cite{Dupire}, typically affected by numerical instabilities in the region $|y|\gg1$.
The use of \eqref{e:steinSteinWings}, together with the evaluation of $\sigma_{\mathrm{loc}}(t,y)$ in a region where numerics can be trusted (say a fixed, or adaptive, bounded domain in the $(t,y)$-plane) leads to a robust and globally defined local volatility surface, that can then serve as the basis for a Monte-Carlo evaluation of exotic option prices, with important consequences for model risk management.
An analogous result for the asymptotic slopes of the local variance in the Heston model \cite{Hest} was given in \cite{DmFG} (where the result \eqref{e:steinSteinWings} for the Stein-Stein model was announced), basing on previous work in \cite{FrGer}.
Note that the analysis in \cite{DmFG} is based on an implementation of the saddle-point method, and is hence tied to the manipulation of characteristic functions and to the affine structure enjoyed by the Heston model \cite{KellRes}.
Our main Theorem \ref{t:conditionalLawStochVol} does not rely on any particular structural assumption on the coefficients of the SDE, and is potentially applicable to families of models larger than the affine class.}
\end{comment}

\begin{proof}[Proof of Theorem \ref{r:steinSteinWings}]
Setting $\eps^2 = 1/|y|$ and using the change of variables $Y^{\eps}_t = \eps^2 Y_t, Z^{\eps}_t = \eps Z_t$ that leads to \eqref{e:steinSteinSmallNoise}, we have the identity
\[
\lim_{y\rightarrow \pm \infty} \frac{ \sigma_{\mathrm{loc}}^{2}(t,y) }{|y|}
=\lim_{y\rightarrow \pm \infty} \frac1{|y|} \mathbb{E} \left[ Z_{t}^{2}| Y_{t}=y\right]
%\\
=\lim_{\eps \rightarrow 0} \mathbb{E}\left[ \left(Z_{t}^{\eps}\right)^{2}| Y_{t}^{\eps}=\pm 1 \right].
\]
The last limit above exists, and is equal to the right hand side of \eqref{e:steinSteinWings}, if the application of Theorem \ref{t:conditionalLawStochVol} is justified.
The functions $\beta_{\eps}(z)=a \eps + b z \to bz =:\beta_0(z)$ and $\alpha(z)=c$ clearly satisfy assumption (SV).
Since the starting point is $(y_0,z_0)=(0,0)$ and the arrival subspaces are $N_{\pm1}=(\pm 1, \cdot)$, we are in case $(ii)$ of  Theorem \ref{t:conditionalLawStochVol}, and the claim follows.
\end{proof}

\subsection{Consistency with the Heston model and moment explosion} \label{s:consistencyHest}

Some basic It\^o calculus shows that when $a=0$, the Stein--Stein model \eqref{e:steinStein} can be obtained as an instance of the Heston model .
More precisely, consider a Heston model for the couple log-price/instantaneous variance $(Y^H_t, V_t)$:
\be \label{e:Heston}
\begin{aligned} 
dY^H_{t} &= -\frac12 V_{t} dt + \sqrt{V_t} d\tilde{B}^{1}_{t}, \qquad Y^H_{0}=0,
\\
dV_{t} &= (q + \kappa V_{t})dt + \xi \sqrt{V_t} d\tilde{B}^{2}_{t}, \qquad V_0=v_0,
\end{aligned}
\ee
where $\kappa<0$; $\xi,v_0>0$ and $\tilde{B}^{1}, \tilde{B}^{2}$ are two Brownian motions with correlation $\rho$.
When the parameters of the Heston model are given by 
\be \label{e:StStHestIntersect}
q=c^2, \quad \kappa=2b, \quad
\xi=2c, \quad v_0=z_0^2,
\ee
then the couple $(Y_t,Z_t^2)$ has same law as $(Y^H_t,V_t)$, for every $t>0$ (the identity in law actually holds for the entire processes; see \cite[Section 2.4]{ArchilBook} for details).
It follows that $\sigma_{\mathrm{loc}}(t,y)^2=\esp[Z_t^2|Y_t=y]=\esp[V_t|Y^H_t=y]=:\sigma_{\mathrm{loc}}^{Heston}(t,y)^2$ under the particular parameter configuration \eqref{e:StStHestIntersect}.

As pointed out in Comment \ref{c:commentLocalVolWings} above, the local variance is know to be asymptotically linear also in the Heston model (with general parameters); from \cite[Theorem 1]{DmFG}:
\be \label{e:hestonLocalVar}
\lim_{y \to \pm\infty} \frac{\sigma^{Heston}_{\mathrm{loc}}(t,y)^2}{|y|} = \frac{2 R_2(s_\pm)}{s_\pm(s_\pm-1)R_1(s_\pm)},
\ee
where $s_+ = s_+(t):=\sup\{s>0: \esp[e^{sY^{H}_t}]<\infty \}$ and $s_- = s_-(t):=\inf\{s<0: \esp[e^{sY^{H}_t}]<\infty \}$ are the upper resp. the lower critical exponents of $Y_t$, and
\begin{align}\label{eq:R1}
R_{1}(s)& = T c^{2}s(s-1)\left[ c^{2}(2s-1)-2\rho c(s\rho c+b)%
\right]  \\
& \quad -2(s\rho c+b)\left[ c^{2}(2s-1)-2\rho c(s\rho c+b)\right]
\notag \\
& \quad +4\rho c\left[ c^{2}s(s-1)-(s\rho c+b)^{2}\right];
\notag \\
R_{2}(s)& =2c^{2}s(s-1)\left[ c^{2}s(s-1)-(s\rho c+b)^{2}%
\right].  \label{eq:R2}
\end{align}
Under the parameter configuration \eqref{e:StStHestIntersect}, the asymptotic behavior \eqref{e:hestonLocalVar} should be consistent with Theorem \ref{r:steinSteinWings}; that is, the two limits should have the same value.
Since the two expressions on the right hand sides of \eqref{e:steinSteinWings} and \eqref{e:hestonLocalVar} are hardly assessed at a glance, we carefully check below - as a `sanity check' for our Theorem \ref{r:steinSteinWings} - that the two constants obtained by different methods are indeed equal.
Let us choose the $+$ sign in \eqref{e:hestonLocalVar} and \eqref{e:steinSteinWings} in what follows; the case with the $-$ sign is handled analogously.
\medskip

Under the condition $\kappa<0$ and $\rho \le 0$ in \eqref{e:Heston}, the critical exponents $s_+$ for the Heston model is the positive solution of the equation
\be \label{e:explTimeHeston}
T^*(s) := \frac2{{\sqrt{-\Delta(s)}}}\left(\arctan \frac{\sqrt{-\Delta(s)}}{\kappa+\rho \xi s} + \pi \right) = t
\ee
where $T^*(s)=\sup\{T>0: \esp[e^{sY^{H}_T}]<\infty\}$ is the explosion time of the (exponential) moment of order $s$, with
\[
-\Delta(s) := \xi^2 s(s-1) - (\kappa+\rho \xi s)^2.
\]
The explosion time $T^*$ can be computed explicitly exploiting the fact that the Heston couple $(Y^H,V)$ is an affine process; see \cite{KellRes}.
Under the condition $\rho^2 \neq 1$, $-\Delta(s)$ is positive for $s$ larger than its positive root $s_2$.
It is not difficult to see, then, that \eqref{e:explTimeHeston} admits a unique positive solution $s_+(t) \in (s_2,+\infty)$ for any value of $t>0$.\footnote{
In its turn, the negative exponent $s_-(t)$ is the unique solution of equation \eqref{e:explTimeHeston} on $(-\infty,s_1)$, where $s_1$ is the negative root of $-\Delta(s)$.
}

\begin{lemma} \label{l:mappingRoots}
Define the function $R(s) := \frac t2 \sqrt{-\Delta(s)}$ for $s>s_2$, where $s_2$ is the positive root of $-\Delta(s)$.
Denote
\be \label{e:sOfR}
s(R) :=
\frac1{2(1-\rho^2)}
\left[ 1+2\rho \frac{\kappa}{\xi} + \sqrt{(1+2\rho\frac{\kappa}{\xi})^2+4(1-\rho^2)\left(\frac{\kappa^2}{\xi^2}+\frac{4R^2}{\xi^2 t^2}\right)}
\right]
\ee
the inverse of the map $s \mapsto R(s)$.
Then, $s_+$ is the unique positive solution of equation \eqref{e:explTimeHeston} if and only if $R(s_+)=\frac t2\sqrt{-\Delta(s_{+})}$ is the smallest positive solution of the equation
\be \label{e:root2}
R \cos R - \frac{t}2 (\kappa+\rho \xi s(R)) \sin R = 0.
\ee
\end{lemma}

\begin{proof}
For every $k \in \mathbb{N}^*$, the equation
\be \label{e:explosionGenericEq} 
\arctan\left(\frac{2R}{t(\kappa+\rho \xi s(R))}\right) + k \pi = R
\ee
has a unique root $R_k > 0$.
Applying the tangent function to both sides, it is seen that $\{R_k: k \in \mathbb{N}^*\}$ coincides with the set of (infinitely many) solutions to the equation $\frac{2R}{t (\kappa+\rho \xi s(R))}=\tan(R)$, which is equation \eqref{e:root2}.
Using $\kappa<0$, $\rho \le 0$ and $s(R) \ge 0$, it is easy to see that the smallest positive solution to \eqref{e:root2} is contained in the interval $(\pi/2,\pi)$.
On the other hand, using $\arctan \in (-\pi/2,\pi/2)$, it is clear that $R_1 \in (\pi/2,\pi)$ while $R_k \notin (\pi/2,\pi)$ for $k >1$; it then follows that the smallest positive root of \eqref{e:root2} coincides with the unique root $R_1$ of \eqref{e:explosionGenericEq} with $k=1$.
Setting $s_1=s(R_1)$, i.e. $R_1=\frac t2\sqrt{-\Delta(s_1)}$, it then holds that $s_1$ is the unique positive solution to
\[
\arctan\left(\frac{\sqrt{-\Delta(s)}}{t(\kappa+\rho \xi s)}\right) + \pi = \frac t2 \sqrt{-\Delta(s)}
\]
which is equation \eqref{e:explTimeHeston}, therefore $s_1 = s_+$.
Conversely, if $s_+$ denotes the unique root of \eqref{e:explTimeHeston}, then $R(s_+)=R_1$, and the claim is proved.
\end{proof}
\medskip

Now consider the particular Heston parameterization in \eqref{e:StStHestIntersect}. 
Plugging $\kappa=2b$ and $\xi=2c$ into \eqref{e:sOfR} shows that the function $r \mapsto s(r)$ in Lemma \ref{l:mappingRoots} coincides with the function $r \mapsto p(2r,1)$, with $p$ defined in \eqref{e:p}.
Then comparing equations \eqref{e:rRoot} and \eqref{e:root2}, it follows from Lemma \ref{l:mappingRoots} that
\be \label{e:sAndR}
2R(s_+) = t\sqrt{-\Delta(s_+)}
= t\sqrt{c^2 s_+(s_+-1)-(b+\rho c s_+)^2}
= \overline{r}(1),
\ee
or yet $s_+=s(\overline{r}(1)/2)=p(\overline{r}(1),1)$.
\bigskip

\textbf{The two constants obtained by different methods are the same}.
Denote
\[
A_{Hest}^2 := \frac{2 R_2(s_+)}{s_+(s_+-1)\overline{r}(s_+)};
\qquad
A_{StSt}^2 := \Bigl( \frac{q c^2 t}{\overline{r}} \sin(\overline{r})\Bigr)^2
\]
the asymptotic slopes of local variance as defined in \eqref{e:hestonLocalVar} and \eqref{e:steinSteinWings}, where $\overline{r}=\overline{r}(1)$ and $q=q(1)$.
Plugging $\kappa=2b$ and $\xi=2c$ into $A_H^2$, after some simplifications one obtains
\[
A_{Hest}^2 =
\frac{4 [c^2 s_+(s_+ -1)-\bar{b}^2]}
{T s_+(s_+ -1)[c^2 (2s_+ -1)-2\rho c \bar{b}] - (2s_+ -1)\bar{b} +2\rho c s_+(s_+ -1)}
\]
with $\bar{b}=b+\rho c s_+$.
Substituting $p=p(\overline{r},1)=s_+$ inside the expression for $q$ in \eqref{e:q}, $A^2_{StSt}$ reads
\be \label{e:aStStv2}
\begin{aligned}
A_{StSt}^2
&=
4 c^2 \overline{r}^2
\frac{\sin(\overline{r})^2}
{t(c^2(2s_+ -1)-2\rho c \bar b) (\overline{r}^2-\overline{r}\sin(\overline{r})\cos(\overline{r}))+ 2\rho c \overline{r}^2 \sin(\overline{r})^2}
\\
&=
4 c^2 \overline{r}^2
\left[
(t c^2(2s_+ -1)-2\rho c \bar b) \frac{\overline{r}^2-\overline{r}\sin(\overline{r})\cos(\overline{r})}{\sin(\overline{r})^2}+ 2\rho c \overline{r}^2 \right]^{-1}
\end{aligned}
\ee
Repeatedly applying \eqref{e:rRoot}, one has
\be \label{e:trigonTransform}
\begin{aligned}
\frac{\overline{r}^2-\overline{r} \sin(\overline{r})\cos(\overline{r})}
{\sin(\overline{r})^2}
&= \left(\frac{\overline{r}}{\sin(\overline{r})}\right)^2
- \frac{\overline{r} \cos(\overline{r})}{\sin(\overline{r})}
= \left(\frac{t \bar b}{\cos(\overline{r})}\right)^2
- t \bar b
\\
&= (t \bar b)^2 \left(1+\frac{\overline{r}^2}{(t \bar b)^2}\right)
- t \bar b
= (t \bar b)^2 + \overline{r}^2 -t\bar b
\\
&= t^2 c^2 s_+(s_+-1) -t\bar b
\end{aligned}
\ee
where we have used the identity \eqref{e:sAndR} in the last step.
Using \eqref{e:trigonTransform} and again $\overline{r}^2=t^2 c^2 s_+(s_+-1) - t^2 \bar b^2$ from \eqref{e:sAndR}, after some straightforward simplifications it follows from \eqref{e:aStStv2} that
\[
\begin{aligned}
A_{StSt}^2
&= \frac{4 c^2 \overline{r}^2}
{t^2 c^2 \Bigl[t s_+(s_+-1)(c^2(2 s_+-1)-2 \rho \bar b) - (2 s_+-1) \bar b + 2 \rho c s_+(s_+-1)\Bigr]}
\\
&= \frac{4 [c^2 s_+(s_+-1) - \bar b^2]}
{t s_+(s_+-1)(c^2(2 s_+-1)-2 \rho \bar b) - (2 s_+-1) \bar b + 2 \rho c s_+(s_+-1)},
\end{aligned}
\] 
and the proof that $A_{StSt}^2=A_{Hest}^2$ is complete.

We stress that our proof of consistency of the two local variance slopes $A_{Hest}^2$ and $A_{StSt}^2$ is valid for negative, non-zero correlation.
In the context of implied volatility expansions, a similar consistency result was obtained by \cite[p. 187-192]{ArchilBook}, in the zero-correlation case.

\subsection{Numerical tests}

In the Heston model, the local volatility $\sigma^{Heston}_{\mathrm{loc}}(t,y)$ can be evaluated using the classical inversion of characteristic functions within the computation of Call price derivatives in the Dupire's formula $\sigma_{\mathrm{loc}}(t,y)^2 = \left. \frac{\partial_t C(t,K)}{\frac12 K^2 \partial^2_{KK}C(t,K)} \right|_{K=S_0 e^y}$.
This gives a way of computing the local volatility in the Stein--Stein model with $a=0$, simply coinciding with a Heston local volatility, when the Heston parameters are given by \eqref{e:StStHestIntersect}.
In Figures \ref{f:fig1} and \ref{f:fig2}, we invert the Heston characteristic functions after a shift of the integration contour in the complex plane into an appropriate saddle-point, following the procedure described in \cite{DmFG}, in order to obtain a stable implementation of the local volatility function for large values of $|y|$.
The two figures illustrate the convergence result in Theorem \ref{r:steinSteinWings}, for the two regions $y<0$ and $y>0$: the blue line represents the ratio $y \mapsto \frac{\sigma_{\mathrm{loc}}^2(y,t)}{y \times (z_t(\pm 1))^2}$, which must tend to $1$ as $y \to \pm \infty$.
The empirical asymptotic behavior is in good agreement with formula \eqref{e:steinSteinWings}, for both the wings: as one expects for a space-asymptotic result, the speed of convergence worsens with increasing maturity (i.e. as the density of the process gradually spreads out). 
\\
Note that the adaptive shift of the integration contour into the saddle-point allows to efficiently evaluate $\sigma_{\mathrm{loc}}^2(y,t)$ for large values of $|y|$, but while this procedure is (i) relatively time consuming in comparison to any explicit formula, in particular if used on-the-fly inside a Monte-Carlo simulation and (ii) limited to models allowing for an explicit evaluation of Fourier transforms, the analysis behind Theorem \ref{r:steinSteinWings} can be extended to other models.

\begin{figure}[t]
     \begin{center}
        \subfigure[$t=0.25$]{
%            \label{fig:smileApproxT025NoHeston}
            \includegraphics[width=0.3\textwidth]{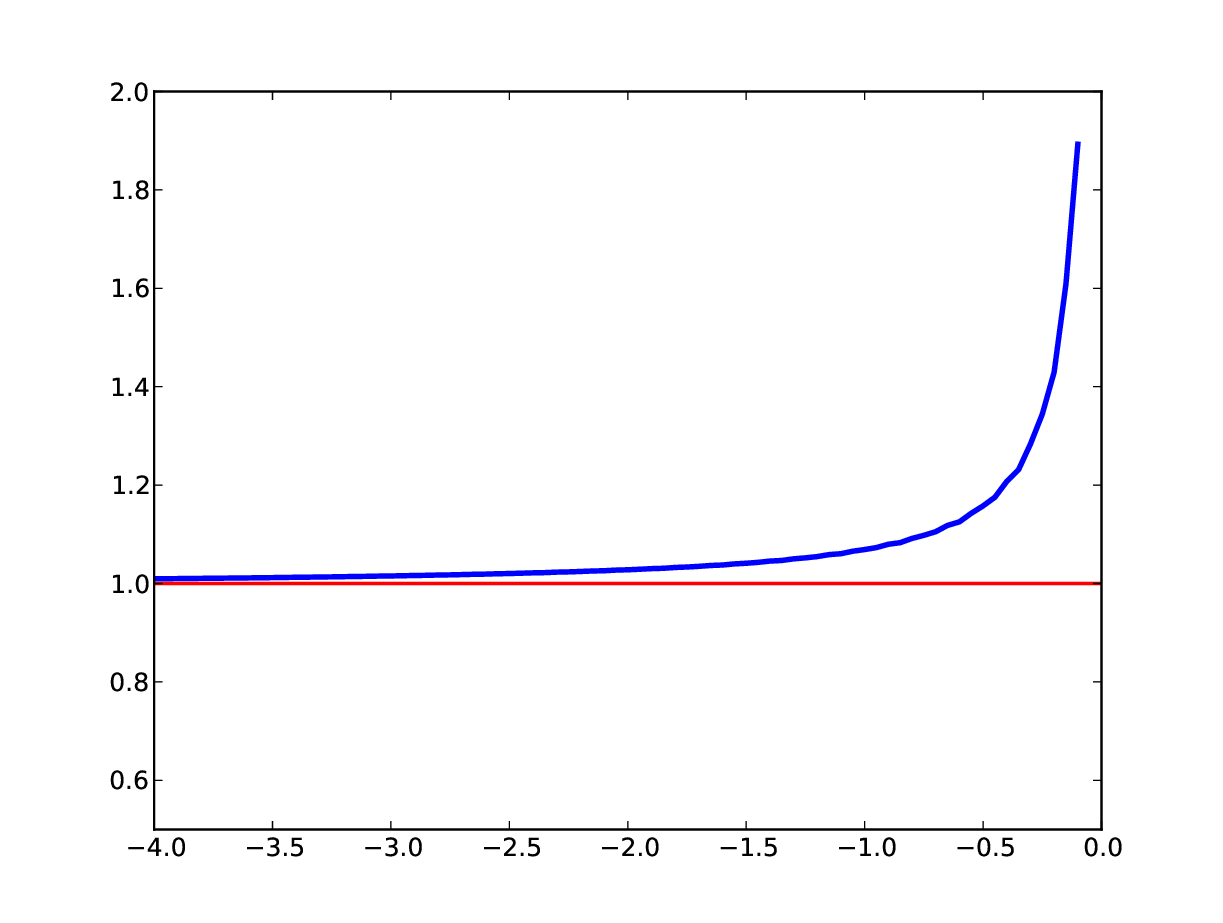}
        }
        \subfigure[$t=1$]{
%           \label{fig:smileApproxT01NoHeston}
           \includegraphics[width=0.3\textwidth]{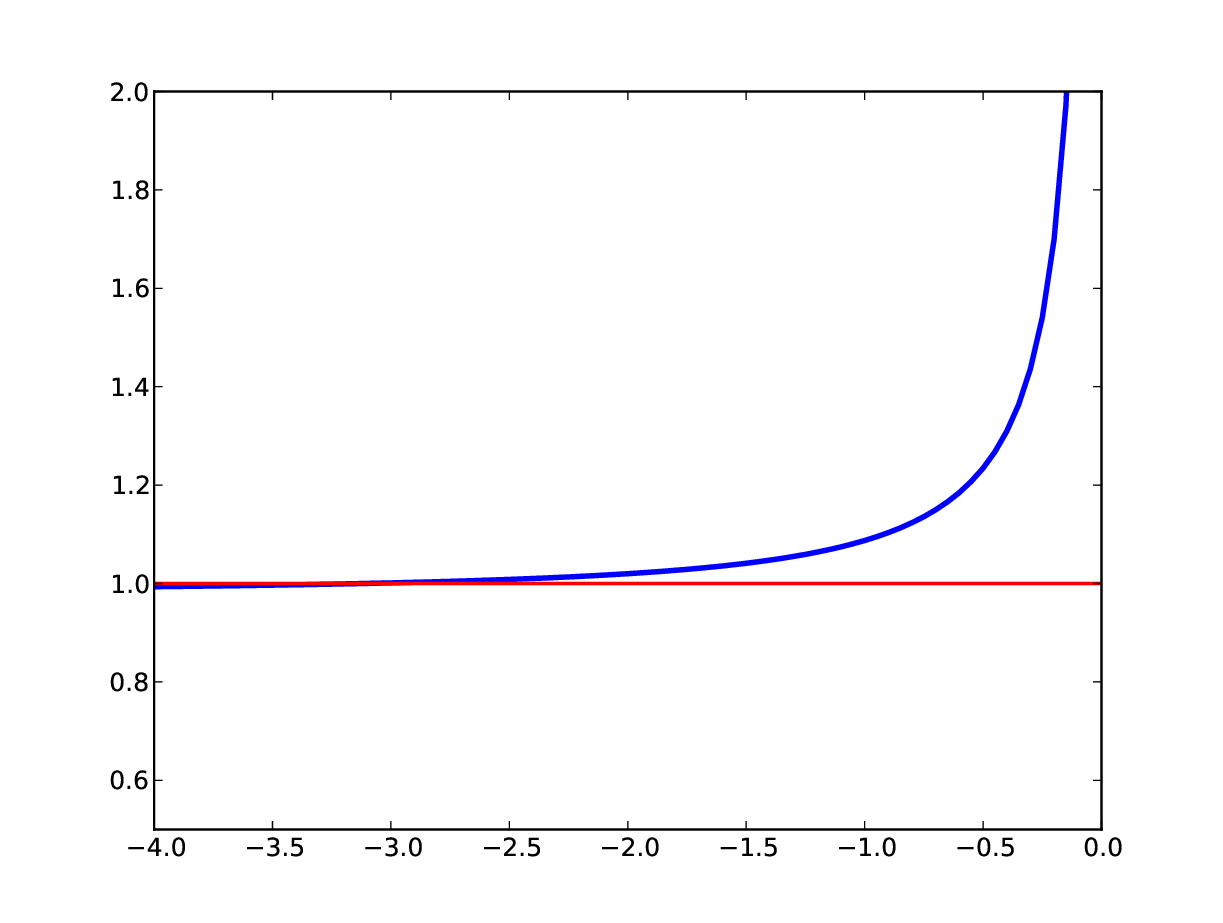}
        }\\ %  ------- End of the first row ----------------------%
    \end{center}
     \begin{center}
        \subfigure[$t=3$]{
%            \label{fig:smileApproxT025NoHeston}
            \includegraphics[width=0.3\textwidth]{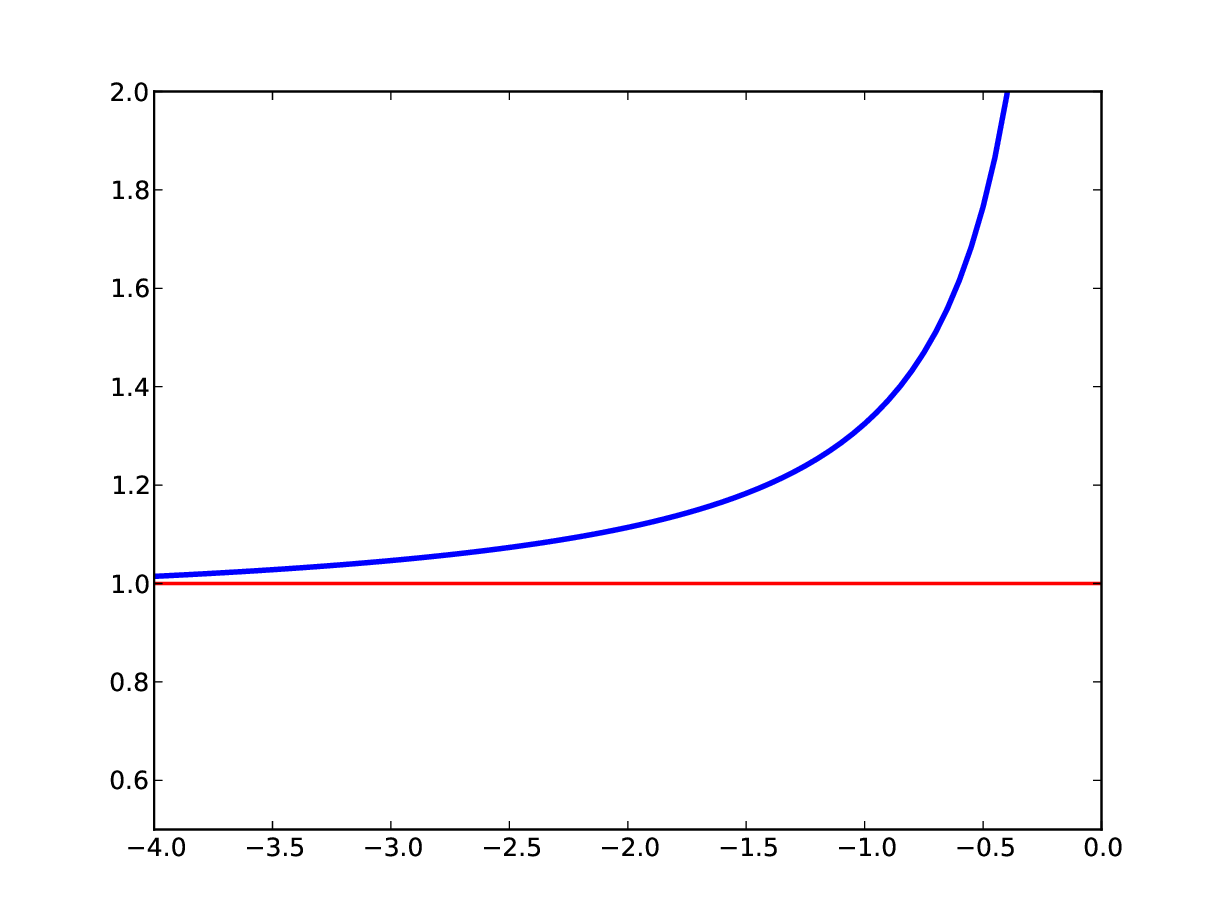}
        }
        \subfigure[$t=10$]{
%           \label{fig:smileApproxT01NoHeston}
           \includegraphics[width=0.3\textwidth]{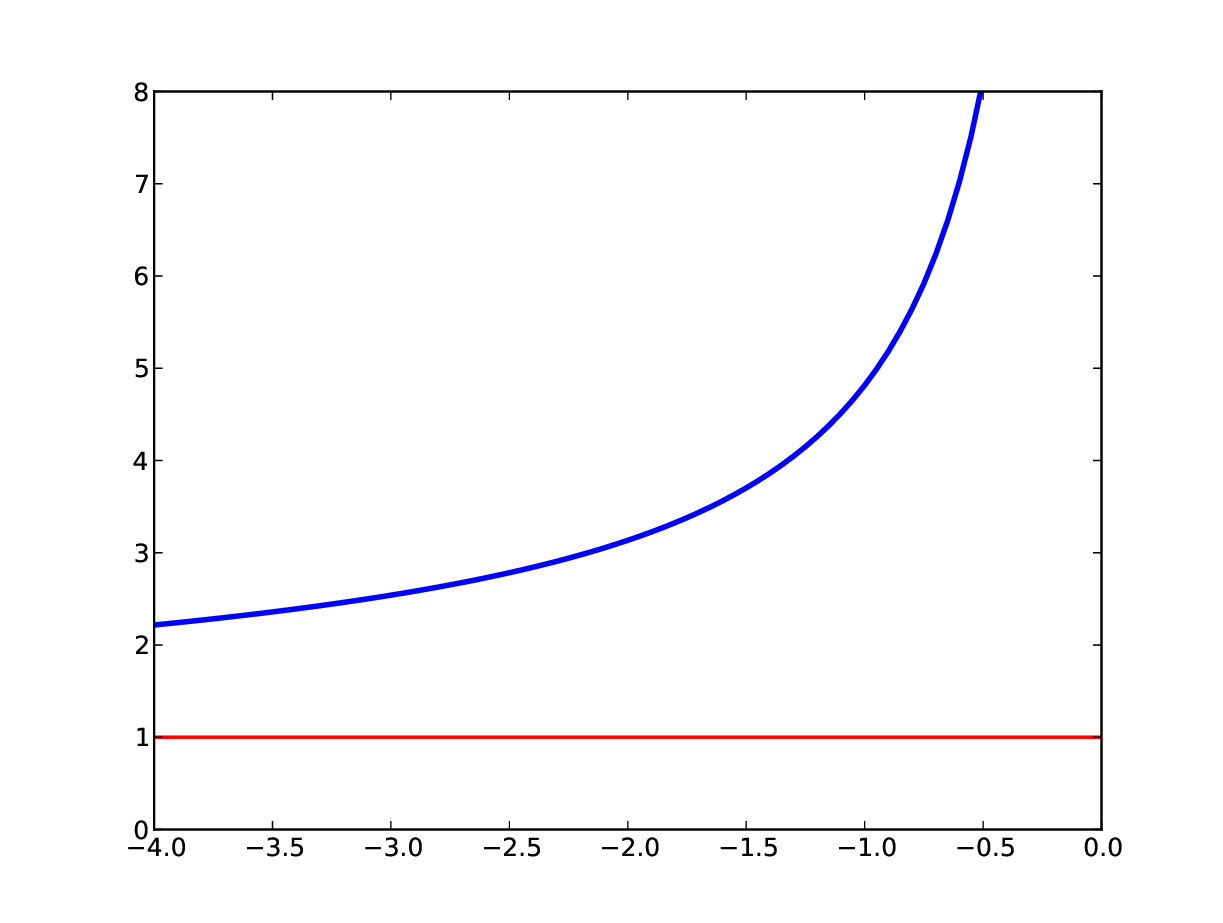}
        }\\ %  ------- End of the second row ----------------------%
    \end{center}
		\caption{Illustration of the convergence result in Theorem
		\ref{r:steinSteinWings} for the Stein--Stein model in the case $y<0$ (left `wing' of the local volatility). The blue line shows the value of the function $y \mapsto \sigma_{\text{loc}}^2(y,t) / ( y \times z_t(-1)^2)$, with $z_t(-1)^2$ the theoretical
		asymptotic value in \eqref{e:steinSteinWings}, which must converge to $1$ as $y \to -\infty$.
		Model parameters: $a=0, b=-0.5, c=0.4, z_0=0.244, \rho=-0.75$.}
     \label{f:fig1}
\end{figure}

\begin{figure}[t]
     \begin{center}
        \subfigure[$t=0.25$]{
%            \label{fig:smileApproxT025NoHeston}
            \includegraphics[width=0.3\textwidth]{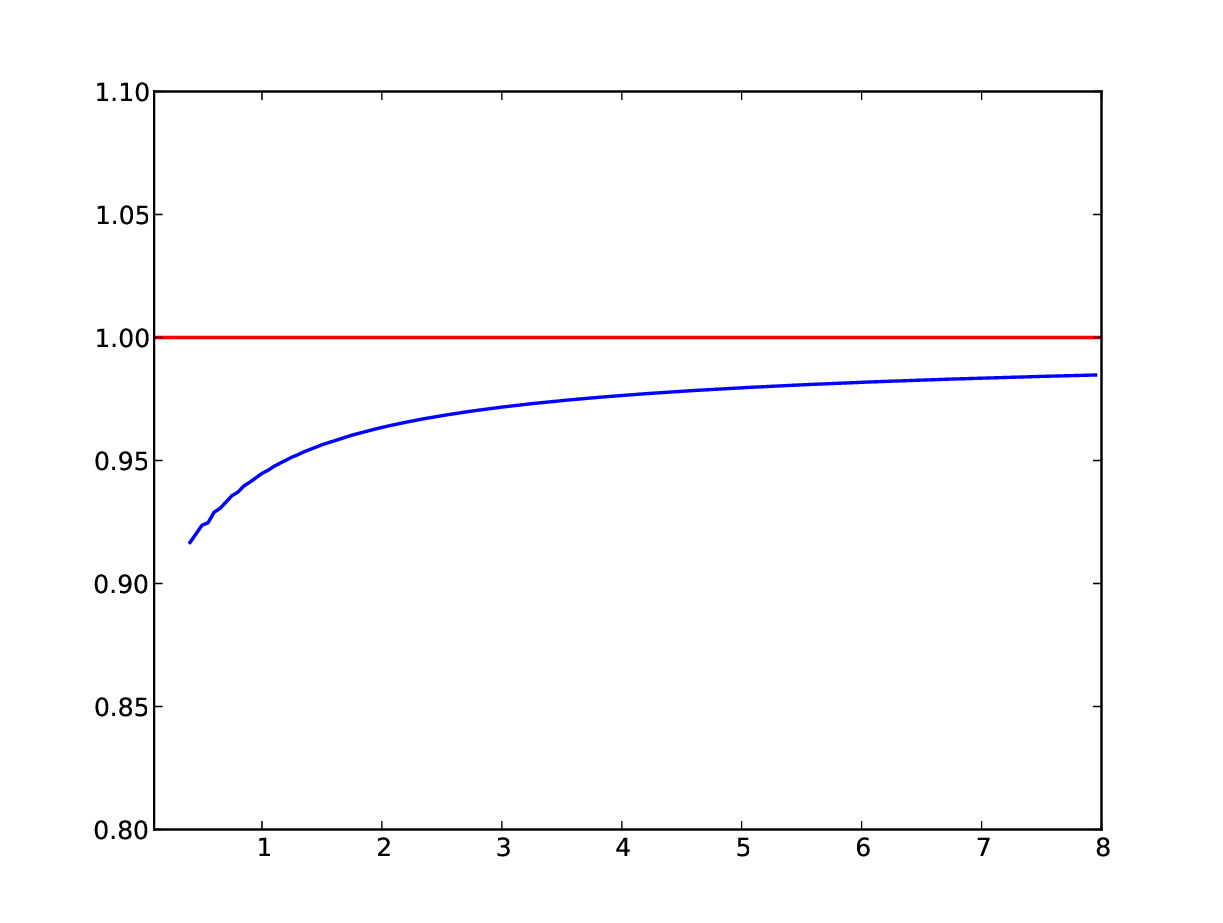}
        }
        \subfigure[$t=1$]{
%           \label{fig:smileApproxT01NoHeston}
           \includegraphics[width=0.3\textwidth]{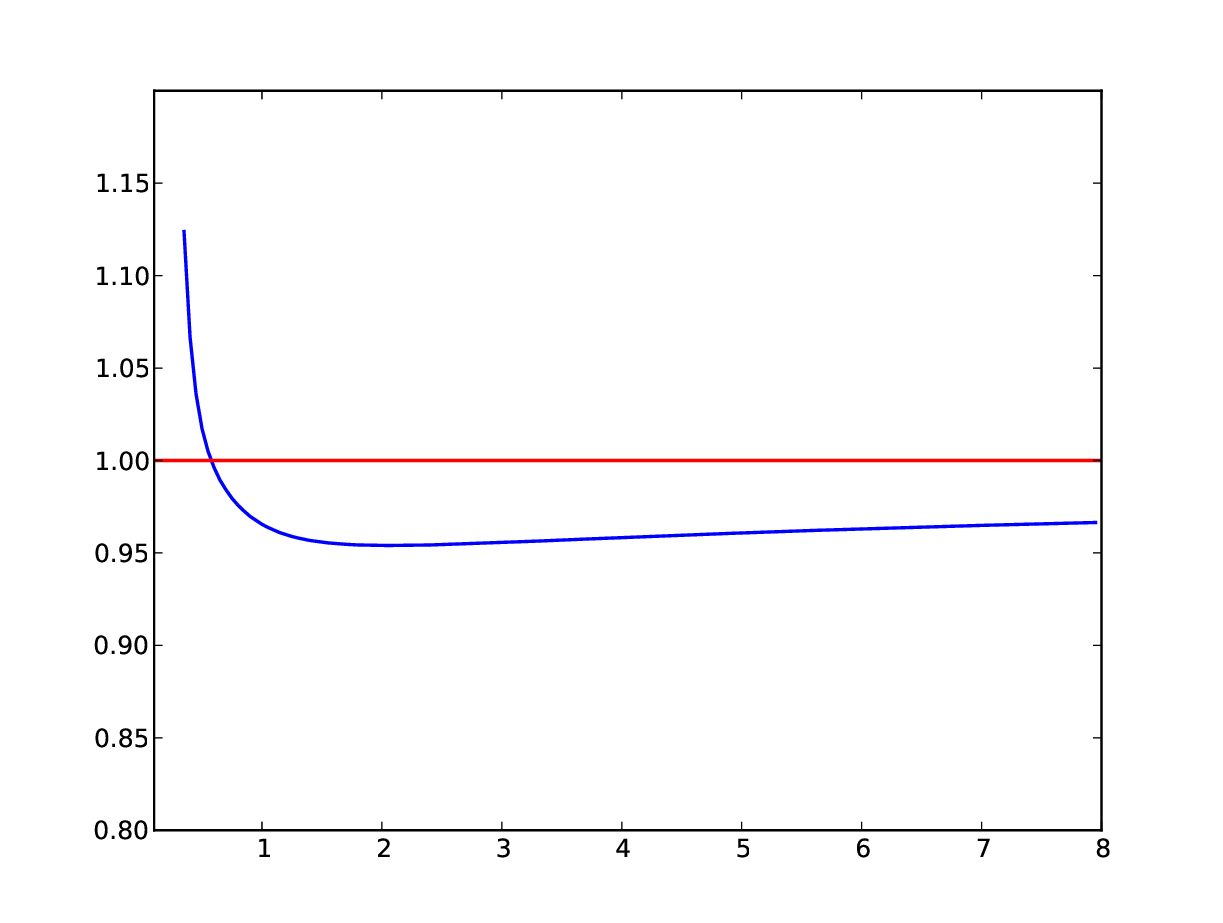}
        }\\ %  ------- End of the first row ----------------------%
    \end{center}
     \begin{center}
        \subfigure[$t=3$]{
%            \label{fig:smileApproxT025NoHeston}
            \includegraphics[width=0.3\textwidth]{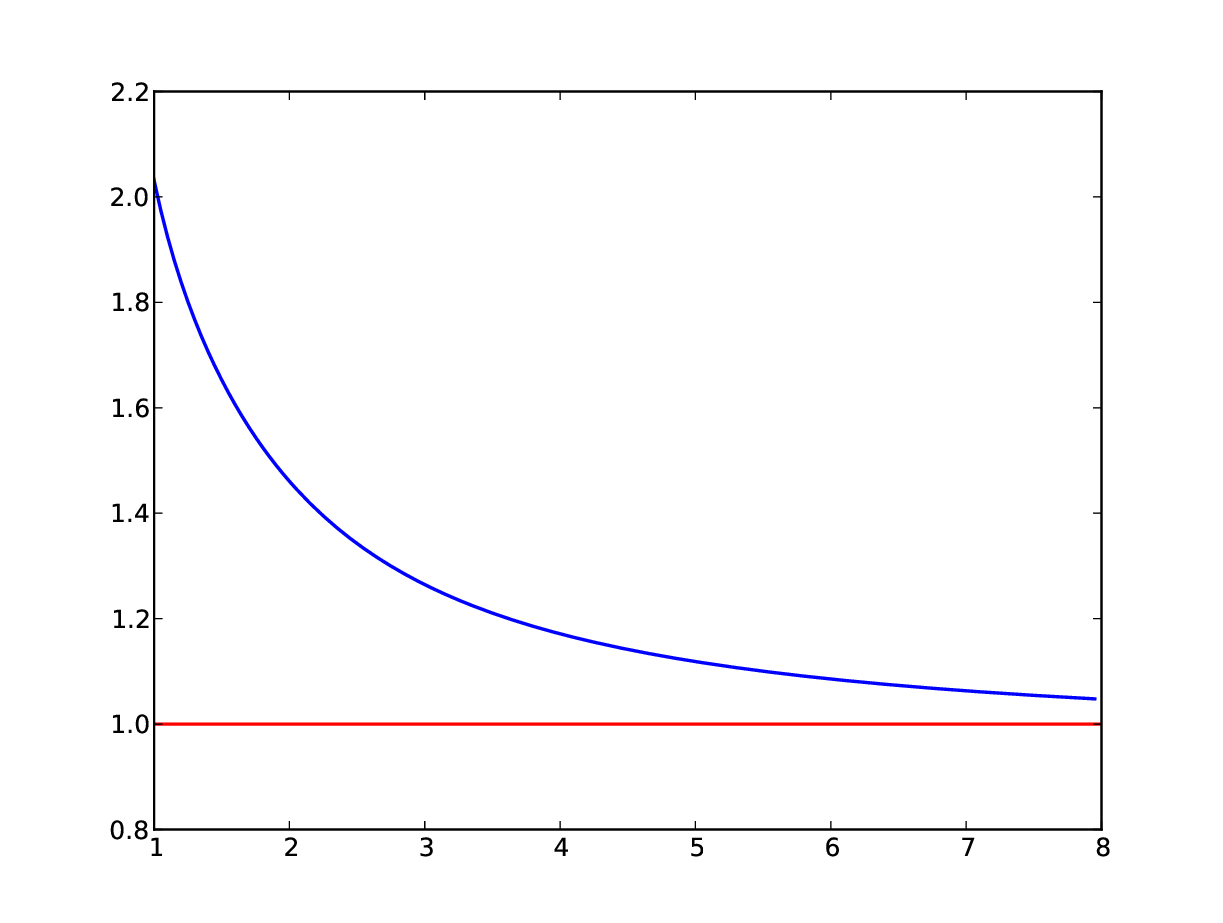}
        }
        \subfigure[$t=10$]{
%           \label{fig:smileApproxT01NoHeston}
           \includegraphics[width=0.3\textwidth]{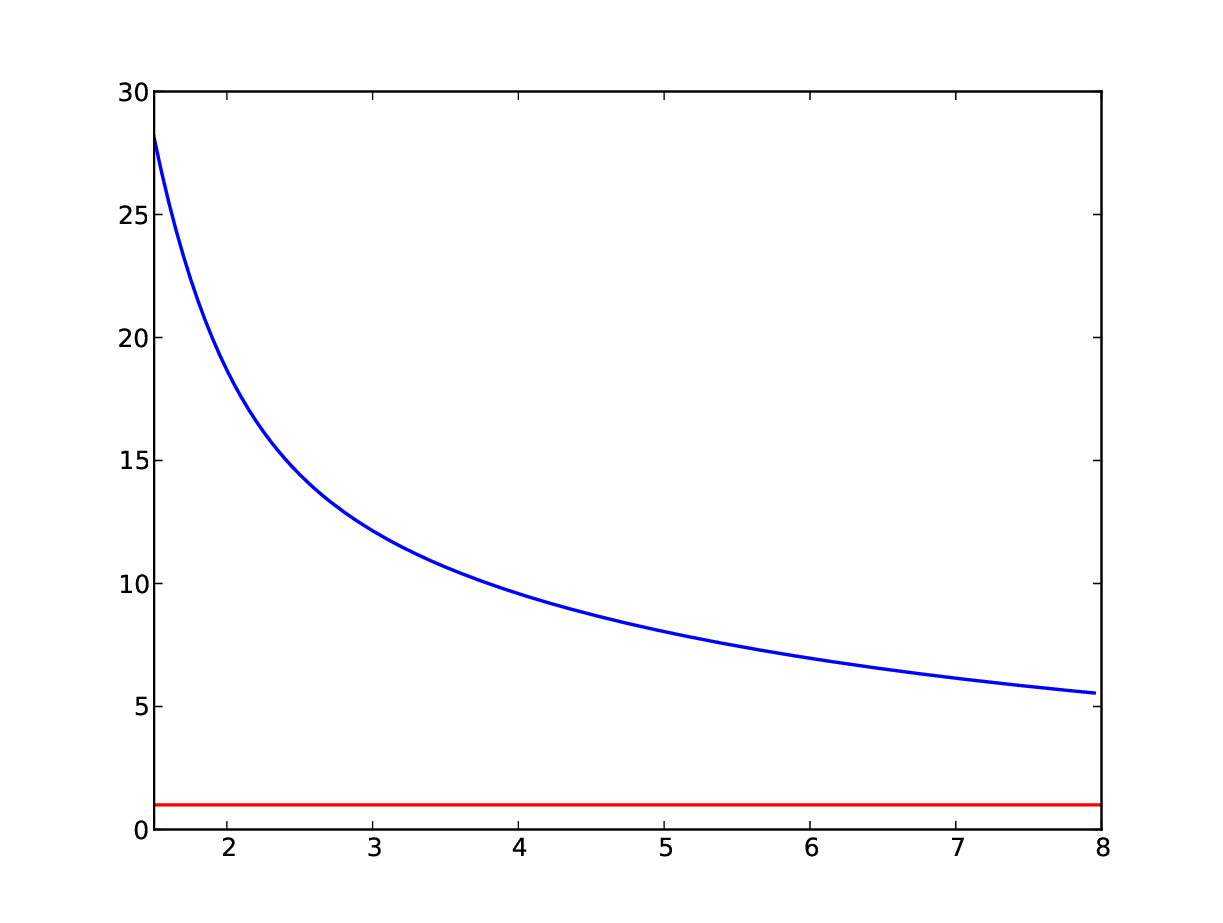}
        }\\ %  ------- End of the second row ----------------------%
    \end{center}
		\caption{
		Illustration of the convergence result in Theorem
		\ref{r:steinSteinWings} for the Stein--Stein model in the case $y>0$ (right `wing' of the local volatility). The blue line shows the
		value of the function $y \mapsto \sigma_{\text{loc}}^2(y,t) / ( y \times z_t(1)^2 )$, with $z_t(1)^2$ the theoretical asymptotic value in \eqref{e:steinSteinWings}, which must converge to $1$ as $y \to \infty$.
		Model parameters: $a=0, b=-0.5, c=0.4, z_0=0.244, \rho=-0.75$.
	  }
    \label{f:fig2}
\end{figure}

\appendix

\section{Technical proofs}

\subsection{Localization} \label{s:app0}

Here we define the precise localization procedure for the SDE \eqref{e:smallNoiseStochVol} that is used in Theorem \ref{t:conditionalLawStochVol}, and state the lemma about the localization of the action that is used therein.
Consider a family of truncation functions $\psi_R \in C^{\infty}_b(\R; \R)$, $R>0$, such that
\[
\psi_R(x) = \left\{\begin{array}{l l} x & \mbox{if } |x| \le R
\\
R+1 &\mbox{if } |x| \ge R+1
\end{array} \right.
\]
and set
\be \label{e:truncated}
\beta_\eps^R(x) = b_{\eps}(\psi_R(x)); \quad
\beta_0^R = \beta_0(\psi_R(x)); \quad
\alpha^R = \alpha(\psi_R(x));
\ee
so that for every $R$, the truncated vector fields $\beta_\eps^R, \beta_0^R$ and $\alpha^R$ coincide with the original ones on the ball $B_R(0)$, but they remain bounded on $\R^n$ (uniformly in $\eps$ in the case of $\beta_\eps^R$, thanks to assumption \eqref{e:driftConvergenceStochVol}).
If $K$ is a common Lipschitz constant for $b_0$ and $\alpha$, then
\be \label{e:truncatedLip}
Lip(b_0^R) \le K; \qquad Lip(\alpha^R) \le K, \qquad \mbox{for all } R
\ee
and it is clear that $\alpha^R(0)=\alpha(0) \neq 0$ under assumption (SV).
It follows that strong \Horm condition (sH) holds at all points also for the controlled system with truncated coefficients
\be \label{e:detFlowTruncated}
\begin{aligned}
dy^R_t(h) &= -\frac12 1_{\theta=0} \bigl(z^R_t(h)\bigr)^2 dt + z^R_t(h) d\overline{h}^1_t, \quad y^R_0(h) = 0,
\\
dz^R_t(h) &= \beta_0^R(z^R_t(h) )dt + \alpha^R(z^R_t(h)) d\overline{h}^2_t, \quad z^R_0(h) = z_0
\end{aligned}
\ee
where $\overline{h} = \sqrt{\Gamma} h$.
Denote $\Lambda_t^R(x) = \inf \{\frac12 |h|_H^2: h \in \mathcal{K}_t^x(R) \}$, $x \in \R^2$, the action of the system \eqref{e:detFlowTruncated}, with $\mathcal{K}_t^x(R) = \{ h \in H: (y^R_0,z^R_0)=x_0, (y^R_t,z^R_t) = x \}$.

\begin{lemma} \label{l:rateFctLoc}
Assume $\beta_0$ and $\alpha$ are Lipschitz continuous on $\R^n$.
Denote $\varphi(h)=(y(h),z(h))$ the solution to the ODE \eqref{e:ODEstochVol}, and $\Lambda^{SV}_t$ the related action as in \eqref{e:rateFctXstochVol}.
For every $R>0$, define $b_0^R$ and $\alpha^R$ according to \eqref{e:truncated}, the corresponding ODE solution $\varphi^R(h)=(y^R(h),z^R(h))$ as in \eqref{e:detFlowTruncated}, and the related action $\Lambda_t^R$.
Let $y \in \R^l$ and $t>0$ be fixed as in Theorem \ref{t:conditionalLawStochVol}.
Then, if at least one of Conditions $(i)$ and $(ii)$ of Theorem \ref{t:conditionalLawStochVol} is satisfied:

\begin{itemize}
\item[(a)] for every $R>0$, there exists $\overline{R}(R)>0$ such that
\[
\Lambda^{\overline{R}(R)}_t(y,z) = \Lambda^{SV}_t(y,z) \qquad \forall \: |z| \le R.
\]
\item[(b)] There exists $\tilde{R}>0$ such that the maps $z \mapsto \Lambda^{SV}_t(y,z)$ and $z \mapsto \Lambda_t^{\tilde R}(y,z)$ attain their (common) global minimum at the same points:
\[
\{z \in \R^{n-l}: \Lambda^{SV}_t(y,z)=\Lambda^{SV}_t(N_y) \} = \{z \in \R^{n-l}: \Lambda_t^{\tilde{R}}(y,z)=\Lambda_t^{\tilde{R}}(N_y)
=\Lambda^{SV}_t(N_y) \}.
\]
\end{itemize}
\end{lemma}

\begin{proof}
Arguing as in the proof of Theorem \ref{t:conditionalLawStochVol}, if at least one of Conditions $(i)$ or $(ii)$ is satisfied, the deterministic Malliavin matrix $C_{(0,z_0)}(h)$ is invertible for all $h \in \mathcal{K}_t^{(y,\cdot)}$; it then follows from Lemma \ref{l:propertiesLambda}$(v)$ that $\Lambda_t$ is continuous on an open set of $\R^2$ containing the line $N_y=(y,\cdot)$.

Let us prove $(a)$. Fix $R>0$, and define
$\mathcal{H}_R = \Bigl\{h \in H: \frac12 |h|_H^2 \le \sup_{|x|\le R} \Lambda^{SV}_t(x) \Bigr\}$.
Since $\Lambda_t^{(SV)}$ is finite and continuous around $N_y$, $\sup_{|z|\le R} \Lambda^{SV}_t(y,z)$ is finite and $\mathcal{H}_R$ is bounded in $H$.
Denote $\mathcal{K}^x_{min}$ the set of minimizing controls in $\mathcal{K}^x_t$.
It is clear that, for every $z \in B_R(0)$, $\mathcal{K}^{(y,z)}_{min} \subset \mathcal{H}_R$, hence
$\Lambda^{SV}_t(y,z) = \inf \Bigl\{\frac12 |h|_H^2: h \in \mathcal{K}_t^{(y,z)} \cap \mathcal{H}_R \Bigr\}$, and $\{y\} \times B_R(0) \subset \{\varphi_t(h) : h \in \mathcal{H}_R \}$.
%that is, every $x$ in $B_R(0)$ can be reached following a control in $\mathcal{H}_R$.
Setting
\[
C \exp\Bigl(C \sup_{h \in \mathcal H_R}|h^{(2)}| \Bigr)
\le
C \exp \Bigl(C
\Bigl( 2 \sup_{|z|\le R} \Lambda^{SV}_t(y,z) \Bigr)^{1/2}
\Bigr)
:= \overline{R}(R)
\]
where $C=C(t,z_0,K)$ is the constant defined in Lemma \ref{l:propertiesLambda}$(i)$ and $K$ is a Lipschitz constant for $\beta_0$ and $\alpha$, it follows from estimate \eqref{e:gronwallSolutionMap} that
$\sup_{s \le t} |z_s(h)| \le \overline{R}(R)$ for every $h \in \mathcal{H}_R$.
Therefore, for such $h$ the trajectory $s \mapsto z_s(h)$, $s\le t$, remains in the region where the vector fields $b_0, \sigma_j$ and $b^{\overline{R}(R)}_0, \sigma^{\overline{R}(R)}_j$ coincide: from the uniqueness of solutions for the second ODE in \eqref{e:ODEstochVol}, it follows that $z^{\overline{R}(R)}_s=z_s$ for all $s \in [0,t]$.
Since $y(h)$ only depends on $z(h)$, one also has $y^{\overline{R}(R)}_s(h)=y_s(h)$ for all $s \le t$, hence $\varphi^{\overline{R}(R)}_t(h)=\varphi_t(h)$ for all $h \in \mathcal{H}_R$.
In particular, $\varphi^{\overline{R}(R)}_t(h_0)=\varphi_t(h_0)$ for $h_0 \in \mathcal{K}_{\min}^{(y,z)}$, and this establishes
\be \label{e:localizationIneq:1}
\Lambda_t^{\overline{R}(R)}(y,z) \le 
\frac12 |h_0|^2 = \Lambda^{SV}_t(y,z), \qquad \forall z \in B_R(0).
\ee
On the other hand, $\Lambda_t^{\overline{R}(R)}(y,z) = \frac12|\overline{h}_0|^2$ for some $\overline{h}_0 \in \mathcal{K}^{(y,z)}(\overline{R}(R))$: from \eqref{e:localizationIneq:1}, $\overline{h}_0 \in \mathcal{H}_R$, therefore $\varphi_t(\overline{h}_0)=\varphi^{\overline{R}(R)}_t(\overline{h}_0)=(y,z)$ by the uniqueness argument above.
This implies $\overline{h}_0 \in \mathcal{K}^{(y,z)}$ and $\Lambda^{SV}_t(y,z) \le \frac12 |\overline{h}_0|^2 = \Lambda_t^{\overline{R}(R)}(y,z)$, and $(a)$ is proved.

Let us now prove $(b)$: note that estimate \eqref{e:gronwallSolutionMap} in Lemma \ref{l:propertiesLambda}, together with the bound \eqref{e:truncatedLip} on the Lipschitz constants of $\beta_0^R$ and $\alpha^R$, implies
\be \label{e:inBall}
%\max\Bigl(	
\sup_{s \le t} |z_s(h)| \vee \: \sup_{s \le t} |z^R_s(h)|
%\Bigr)
\le C e^{C |h|_H}, \qquad \forall \: R
\ee
where again $C=C(t,z_0,K)$.
Set $R^{1} = C \exp \Bigl(C \Bigl( 2(\Lambda_t^{SV}(N_y)+1) \Bigr)^{1/2}
\Bigr)$.
Estimate \eqref{e:inBall} yields the inclusion
$\{z \in \R^n: \Lambda_t^{SV}(y,z) \le \Lambda^{SV}_t(N_y)+1 \} \subseteq B_{R^{1}}(0)$,
because for every $z$ on the left hand side there exists $h_0$ with $\frac12|h_0|^2=\Lambda_t^{SV}(y,z)\le \Lambda^{SV}_t(N_y)+1$ such that $z=z_t(h_0)$.
In particular, all the points $z^*$ are contained in $B_{R^{1}}(0)$, where $z^*$ is a global minimizer of the map $z \mapsto \Lambda^{SV}_t(y,z)$.
Now setting $\tilde{R}=\overline{R}(R^1)$, point $(a)$ of the current lemma entails that $\Lambda^{\tilde{R}}_t$ and $\Lambda^{SV}_t$ coincide on $\{y\} \times B_{R^{1}}(0)$, therefore they - trivially - attain their common minimum $\Lambda^{SV}_t(N_y)$ on $\{y\} \times B_{R^{1}}(0)$ at the same points.
One has to make sure that $\Lambda^{\tilde{R}}$ is strictly above the level $\Lambda^{SV}_t(N_y)$ outside $\{y\} \times B_{R^{1}}(0)$: but \eqref{e:inBall} ensures that $z^{\tilde{R}}_t=z \notin B_{R^{1}}(0)$ can be reached only along controls $h$ such that $\frac12 |h|_H^2 > \Lambda_t^{SV}(N_y)+1$, hence $\Lambda^{\tilde{R}}_t(y,z) \ge \Lambda_t^{SV}(N_y)+1$ for $z \notin B_{R^{1}}(0)$, and this allows to conclude.
\end{proof}

\subsection{Proofs of Theorem \ref{t:BAL} and Proposition \ref{p:tailInt}} \label{s:app2}
 
The following lemma plays a key role.
The proof is given under weak H\"ormander condition at the starting point $x_0$:
\be \label{e:weakH}
\begin{aligned}
\text{(wH)} \qquad \text{span} \{ \sigma_1, \dots,\sigma_d ; \ &[b_0, \sigma_i] : 1 \le i \le d; \
[\sigma_i, \sigma_j] : 1 \le i,j \le d;
\\
&[[\sigma_i, [\sigma_l, \sigma_m]] : 1 \le i,l,m \le d; \dots \} \bigr|_{x_0} = \R^n
\end{aligned}
\ee
that is, the linear span of $\sigma_1, \dots, \sigma_d$ and all the Lie brackets of $b_0, \sigma_1,
\dots, \sigma_d$ is the whole $\R^n$ at $x_0$.
Notice that under assumptions \eqref{e:initCond} and \eqref{e:convergDrift}, one has
\[
[b_\eps, \sigma_j]_{x_0^{\eps}} = [b_0, \sigma_j]_{x_0} + o(1), \qquad
[\sigma_j, \sigma_k]_{x_0^{\eps}} = [\sigma_j, \sigma_k]_{x_0} + o(1)
\]
hence (wH) also holds when $b_0$ and $x_0$ are replaced with $b_{\eps}$ and $x_0^{\eps}$, when $\eps$ is small enough.
It is then classical that $X^{\eps}_t$ admits a smooth density $p_t^{\eps}$ for all $t > 0$ under condition (wH).

\begin{lemma} \label{l:keyEstim}
Let $X^{\eps}$ be the solution to \eqref{e:baseSDE}.
Assume weak \Horm condition \emph{(wH)} at $x_0$.
Then, for any $t >0$, any $x \in \R^n$ and $q \in (0,1)$ there exist a constant $C_q$ and positive integers $N_1(q)=N_1(q,n)$
and $N(q)=N(q,n,x_0)$ such that
\begin{equation} \label{e:keyEstim}
p^{\eps}_t(x) \le C_q \bigl(1 + R^{-N_1(q)} \bigl)
\ 
\eps^{-N(q)} \ \Prob( |X^{\eps}_t - x| \le R )^q
\end{equation}
for every $R > 0$.
The constant $C_q$ also depends on $t$ and on the bounds on the derivatives of $b_0$ and the $(\sigma_j)_j$.
\end{lemma}
\begin{proof} Given in section \ref{s:lemmaProof}. \end{proof}
\bigskip

\begin{proof}[Proof of Theorem \ref{t:BAL}]
Let us write $\Lambda$ for the action function throughout this proof (dropping the fixed index $t$ from the notation).
The following estimate is obtained with standard uniform large deviations arguments (see also L\'eandre \cite[section 2]{LeMaj}):
 for every $R > 0$, one has
\be \label{e:unifLD}
\limsup_{\eps \to 0} \eps^2 \log \Prob( |X^{\eps}_t - x| \le R ) \le -\Lambda(B_{R}(x))
\ee
uniformly over $x$ in compact sets.

Let first prove \eqref{e:BALupPointwise}.
$x$ is fixed; taking $\limsup_{\eps \to 0} \eps^2 \log$ in estimate \eqref{e:keyEstim} and applying \eqref{e:unifLD}, one has
\[
\limsup_{\eps \to 0} \eps^2 \log p^{\eps}_t(x) \le q \limsup_{\eps \to 0} \eps^2 \log \Prob( |X^{\eps}_t - x| \le R )
\le -q \Lambda(B_{R}(x)).
\] 
Now taking the limit $R \to 0$, $\Lambda(B_{R}(x)) = \inf_{y \in B_R(x)} \Lambda(y) \to \Lambda(x)$ by the lower semi-continuity of $\Lambda$.
Since $q<1$ was arbitrary, \eqref{e:BALupPointwise} follows.

Let us now prove \eqref{e:BALup}.
Under the assumption in Theorem \ref{t:BAL}$(ii)$, it follows from Lemma \ref{l:propertiesLambda}$(iii)$ that $\Lambda$ is continuous on an open neighborhood $V$ of $x$, hence uniformly continuous on compact sets contained in $V$.
Fix a compact ball $B \subset V$ and $\delta > 0$.
We can find $R = R_{B,\delta}$ such that the closed $R$-neighborhood of $B$, $B^{R}=\cup_{y \in B} B_R(y)$, is contained in $V$ and moreover $Osc(\Lambda, B_R(y)) \le \delta$ for all $y \in B$.
In particular,
\be \label{e:unifCont}
\Lambda(B_R(y)) = \inf_{z \in B_R(y)} \Lambda(z) \ge \Lambda(y)-\delta
\ee
for all $y \in B$.
Now taking $\limsup_{\eps \to 0} \eps^2 \log$ in estimate \eqref{e:keyEstim}, and applying \eqref{e:unifLD} and \eqref{e:unifCont}, it follows that for any $q \in (0,1)$
\[
\limsup_{\eps \to 0} \eps^2 \log p^{\eps}_t(y) \le q \limsup_{\eps \to 0} \eps^2 \log \Prob( |X^{\eps}_t - y| \le R )
\le -q \Lambda(B_{R}(y)) \le -q (\Lambda(y)-\delta)
\]
where the limit holds uniformly over $y \in B$.
Since $q<1$ and $\delta$ were arbitrary, the right hand side can be improved to $\Lambda(y)$, and \eqref{e:BALup} follows.

The lower bound \eqref{e:BALdown} is actually estimate (3.5) in Ben Arous and L\'eandre \cite[Theorem III.1]{BAL}. 
Their proof can be adapted to the case where $b$ and $x_0$ depend on $\eps$, under the convergence conditions \eqref{e:initCond} et \eqref{e:convergDrift}.
The statement about $\Lambda=\Lambda_{R}$ in Theorem \ref{t:BAL}$(i)$ is obvious from the definitions of the two actions, and \eqref{e:BALcoincide} is
a direct consequence of \eqref{e:BALupPointwise} and \eqref{e:BALdown}.
\end{proof}
\bigskip

\begin{proof}[Proof of Proposition \ref{p:tailInt}]
It follows from Lemma \ref{l:keyEstim} that, for every $q,q' \in (0,1)$ and $R$ small enough,
\begin{multline*}
\int_{|z-\overline{z}| \ge A} |z|^k p^{\eps}_t(y,z) dz \le
C_q  R^{-N_1(q)} \eps^{-N(q)}
\int_{|z-\overline{z}| \ge A} |z|^k \Prob( |X^{\eps}_t - (y,z)| \le R )^q dz
\\
= C_q  R^{-N_1(q)} \eps^{-N(q)}
\int_{|z-\overline{z}| \ge A} |z|^k \Prob( |X^{\eps}_t - (y,z)| \le R )^{qq'} \: \Prob( |X^{\eps}_t - (y,z)| \le R )^{q(1-q')}
dz.
\end{multline*}
Since
\[
\begin{aligned}
\Prob(|X^{\eps}_t - (y,z)| \le R) &\le
 \Prob(|Y^{\eps}_t - y| \le R, |Z^{\eps}_t - z| \le R)
\end{aligned}
\]
one has
\begin{multline} \label{e:integ1}
\int_{|z-\overline{z}| \ge A} |z|^k p^{\eps}_t(y,z) dz \le
C_q \: R^{-N_1(q)} \eps^{-N(q)} \times
\Prob( |Y^{\eps}_t - y| \le R, |Z^{\eps}_t-\overline{z}| \ge A-R)^{qq'}
\\
\times \int_{|z-\overline{z}| \ge A} |z|^k \Prob( |Z^{\eps}_t-z| \le R )^{q(1-q')} dz.
\end{multline}

The integral on the right hand side of \eqref{e:integ1} can be bounded as follows
\be \label{e:tailIntBound}
\int_{|z-\overline{z}| \ge A} |z|^k \Prob( |Z^{\eps}_t-z| \le R )^{q(1-q')} dz
\le 
\int_{\R^{n-l}} |z|^k \Prob( |Z^{\eps}_t| \ge |z|-R )^{q(1-q')} dz.
\ee
Note that the random variable $|Z^{\eps}_t|$ has moment of all orders uniformly bounded in $\eps$ (for so does $|X^{\eps}_t|$):
precisely, for any $T>0$ and $r > 0$ there exists a constant $C_r=C_{r,T}$ such that
$\sup_{\eps \le 1} \: \esp[ (Z^{\eps}_t)^r] \le C_r$ for all $t \le T$.
Then, from Markov's inequality
\begin{equation} \label{e:MarkovZ}
\sup_{\eps \le 1} \: \Prob( |Z^{\eps}_t| \ge |z|-2R ) \le \frac{C_r}{(|z|-R)^r}
\end{equation}
for all $z$ such that $|z|-R>0$ and all $r>0$.
The exponent $r$ can be chosen sucht that
\[
\sup_{\eps \le 1} \Prob( |Z^{\eps}_t| \ge |z|-R )^{q(1-q')}
\le	
\frac{C^{(1)}}{|z|^{k+n-l+1}},
\]
for all $|z|$ larger than, say, $1+R$, for some positive constant $C^{(1)}=C^{(1)}_T$.
It follows that for any choice of $q,q' \in (0,1)$, the integral on the right hand side of \eqref{e:tailIntBound} is convergent, and uniformly bounded in $\eps$.

Finally taking $\log$, multiplying by $\eps^2$ and taking $\limsup_{\eps \to 0}$ in \eqref{e:integ1}, we have
\begin{multline*}
\limsup_{\eps \to 0} \eps^2 \log \int_{|z-\overline{z}| \ge A} |z|^k p^{\eps}_t(y,z) dz \le 
q q' \: \limsup_{\eps \to 0} \eps^2 \log \Prob( |Y^{\eps}_t - y| \le R, |Z^{\eps}_t-\overline{z}| \ge A-R)
\\
\le
-q q' \: \inf \{ \Lambda_t(y',z') : |y'-y| \le R, |z'-\overline{z}| \ge A-R \}.
\end{multline*}
As $R \downarrow 0$, the right hand side tends to $-q q' \: \inf \{ \Lambda_t(y,z'): |z'-\overline{z}| \ge A \}$: since $q, q'<1$ were arbitrary, we obtain the claim.
\end{proof}

\subsection{Proof of Lemma \ref{l:keyEstim}} \label{s:lemmaProof}

Throughout this section, we denote $X = (X_t; t \geq 0)$ the strong solution of the SDE
\be \label{e:diffusion}
X_t =
x_0 + \int_0^t B(X_s)ds + \sum_{j=1}^d \int_0^t A_j(X_s)dW^j_s, \ \ \ t \geq 0
\ee
where $B, A_j \in \mathcal{C}^{\infty}_b (\R^n;\R^n)$ for all $j$.
For any multi-index $\alpha = (\alpha_1, \dots, \alpha_l) \in \{1,\dots,n\}^l$, we denote $|\alpha|=l$ and $\partial_{\alpha} = \partial^{|\alpha|}_{x_{\alpha_1}, \dots, x_{\alpha_l}}$.
Setting
\[
|f|_k = \sum_{|\alpha|\le k} \sup_{y \in \R^n} |\partial_{\alpha} f(y)|
\]
for smooth real valued functions $f$, we denote $\nB_k = 1 + \sum_{i=1}^n |B^i|_k$ and $\nA_k=1 + \sum_{i,j} |A^i_j|_k$.

\emph{Some elements of Malliavin calculus}. Following the standard notation in \cite{Nual06}, we denote $\mathbb{D}^{k,p}$ the domain of the $k$-th order Malliavin derivative, and $\mathbb{D}^{\infty}= \cap_{k \ge 1} \cap_{p \ge 1} \mathbb{D}^{k,p}$.
It is classical, see \cite{Nual06}, that $X_t$ is a smooth random variable in Malliavin's sense for every $t$, that is $X_t \in \mathbb{D}^{\infty}$.
Denoting $D_r X_t=(D_r^1 X_t, \dots, D_r^d X_t)$, $r \in [0,t]$ the ($d$-dimensional) Malliavin derivative of $X_t$, the $k$-th order derivative is obtained by iterating the operator: $D^{j_1, \dots, j_k}_{r_1, \dots r_k} X_t := D_{r_1}^{j_1} \cdots D_{r_k}^{j_k} X_t$, for every $(j_1, \dots, j_k) \in \{1,\dots,d\}^k$.
It is well-known that the random variables $D^{j_1, \dots, j_k}
_{r_1, \dots, r_k} X_t$ have finite moments of any order: the following lemma gives an explicit estimate on the $L^p$ norms, in terms of the bounds on $A$ and $B$ and their derivatives, and will be useful in what follows.

\begin{lemma}[Lemma 2.1 and Corollary 1 in \cite{SDM}] \label{l:MallDerEstim}
For every $k \geq 1$ and $p>1$ there exist positive integers $\gamma, \gamma'$  and a positive constant
$C$, all depending on $k,p$ but not on the bounds on $B$ and $A$ and their derivatives, such that, for
any $t > 0$
\[
\sup_{r_1,\dots,r_k \leq t} \mathbb{E} \left[ \left| D^{j_1,\dots,j_k}_{r_1,\dots,r_k} X \right|^p \right]
\leq
C_{k,p}
\left( t^{1/2} \nB_k + \nA_k \right)^{ \gamma' }
e^{\gamma (t|B|_1 + t^{1/2}|A|_1)^p}
\]
for all $i=1,\dots,m$ and $(j_1, \dots, j_k) \in \{1,\dots,d\}^k$.
Moreover,
\[
|| \phi( X_t ) ||_{k,p} 
\leq C_{k,p} \: |\phi|_k \:
\bigl( 1 + \bigl( t \vee t^k \bigr)^{1/2} \bigr)
\bigl( t^{1/2} \nB_k + \nA_k \bigr)^{\gamma'}
e^{k \gamma (t|B|_1 + t^{1/2}|A|_1)^p}
\]
for any $\phi \in C^{\infty}(\R^n)$.
\end{lemma}

\noindent
The notion of non-degeneracy for (Malliavin-)differentiable random variables $F \in \mathbb{D}^{1,2}$ is understood in the sense of the (stochastic) Malliavin covariance matrix
\[
(\gamma_F)_{i,j} = \int_0^t \sum_{l=1}^d D^l_s F^i \: D^l_s F^j ds,
\quad i,j = 1, \dots, d.
\]
A fundamental tool to study density of random variables with invertible covariance matrix is the integration by parts formula: 

\begin{proposition}[Integration by parts formula of the Malliavin calculus; \cite{Nual06}] \label{p:integrByParts}
Let $F = (F^1, \dots, F^d) \in \mathbb{D}^{\infty}$.
Assume that $\gamma_F$ is invertible a.s. and moreover $\esp [\det(\sigma_F)^{-p}] < \infty $ for all $p \geq 1$.
Let $G \in \mathbb{D}^{\infty}$ and $\phi \in C^{\infty}_{\text{pol}}(\R^m)$.
Then, for any $k \geq 1$ and any multi-index $\alpha = (\alpha_1, \dots, \alpha_k)
\in \{1, \dots, m\}^k$ there exists a random variable $H_{\alpha}(F,G) \in \mathbb{D}^{\infty}$ such that
\be \label{e:integrByParts}
\esp \left[ \partial_{\alpha} \phi(F) G \right] = \esp \left[\phi(F) H_{\alpha}(F,G) \right],
\ee
where the $H_{\alpha}(F,G)$ are recursively defined by 
\[
H_{\alpha}(F,G) = H_{(\alpha_k)}( F, H_{(\alpha_1, \dots, \alpha_{k-1})} (F,G)),
\qquad
H_{(i)}(F,G) = \sum_{j = 1}^m \delta \left( G (\sigma_F^{-1})_{i,j} D F^j \right)
\]
where $\delta$ denotes the adjoint operator of $D$.
\end{proposition}

\noindent
The key ingredient required to apply the integration by parts is an estimate of the $L^p$ norms of the Malliavin weights $H_{\alpha}$.
The following theorem, proved in \cite{SDM}, provides explicit bounds in terms of the bounds on $A$ and $B$ and their derivatives.

\begin{theorem}[Theorem 2.3 in \cite{SDM}] \label{t:MallWeigEstim}
For every $k \geq 1$, there exist a positive constant $C_k$ and positive integers $a_k, b_{k}, c_k$ and $r_k$, all possibly depending also on $n$ and $d$, such that for any multi-index $\alpha \in \{1,\dots,n\}^k$, any $G \in \mathbb{D}^{\infty}$
and any $t > 0$,
\[
||H_{\alpha}(X_t,G)||_2 \le C_k (1 + t^{_k}) \ ||G||_{k,2^{k+1}} \: \esp[ \det (\gamma_t)^{-r_k} ]
\\
\times (|B|_{k+1} + |A|_{k+1})^{b_{k}} e^{(t (|B|_1+|A|_1^2) + t^{1/2} |A|_1)^{c_k}}.
\]
\end{theorem}
\bigskip

Let us go back to equation \eqref{e:diffusion}.
If the stochastic integral in \eqref{e:diffusion} is intended in Stratonovich sense, the drift coefficient $B$ is replaced by $A_0 = B - \frac12 \sum_{j=1}^d \sum_{k=1}^n A^k_j \partial_{x_k} A_j$.
If we assume that the vector field $(A_0,A_j)$ satisfy the weak \Horm condition at $x_0$
\be \label{e:weakHAppendix}
\begin{aligned}
\qquad \text{span} \{ A_1, \dots,A_d ; \ &[A_0, A_i] : 1 \le i \le d; \
[A_i, A_j] : 1 \le i,j \le d;
\\
&[[A_i, [A_l, A_m]] : 1 \le i,l,m \le d; \dots \} \bigr|_{x_0} = \R^n,
\end{aligned}
\ee
then, denoting $V^{L}_{x_0}$ the vector space spanned by the Lie brackets of length smaller or equal to $L$ in \eqref{e:weakHAppendix}, and setting\footnote{The sum in \eqref{e:brackets} is in fact taken over a finite number of generating brackets, and it can alternatively be written using the notation introduced in the Appendix of \cite{KusSt}.}
\be \label{e:brackets}
\mathcal{V}_L (x_0,v; A_0, A) = \sum_{V \in V^{L}_{x_0}} \langle v, V v \rangle^2
\qquad v \in \R^n,
\ee
it follows from \eqref{e:weakHAppendix} that there exists some $L \ge 1$ such that $\mathcal{V}_L(x_0,v; A_0,A)=0 \Rightarrow v=0$.
In other words,
\[
\mathcal{V}_L(x_0;A_0,A) = \inf_{|v|=1} \mathcal{V}_L(x,v; A_0, A) > 0
\]
for some $L \ge 1$.
Under condition (wH) at $x_0$, the Malliavin covariance matrix $\gamma_{X_t}$ satisfies the fundamental estimate of Kusuoka and Stroock \cite[Corollary 3.25]{KusSt}: for every $T>0$ and $r > 0$, there exist a constant $C_r=C_r(T)$ and an integer $N(L,n)$ such that
\be  \label{e:KSMallMatrix}
\esp[ \det (\gamma_t)^{-r} ]^{1/r} \le \frac{C_r}{t^{nL} \ \mathcal{V}_L(x_0;A_0,A)^{N(L,n)}}
\ee
for all $t \in (0,T]$.
\medskip

We are now ready to give the following
\medskip

\begin{proof}[Proof of Lemma \ref{l:keyEstim}]
Let $1_{\{ B_{R/2}(0)\}} \le \overline{\varphi}_R \le 1_{\{ B_R(0) \}}$ be a $C^{\infty}$ function.
We can construct $\overline{\varphi}_R$ so that $|\overline{\varphi}_R|_k \le C_k (1 + R^{-k})$ for some constant $C_k$ (eventually depending on the dimension $n$).
Define $\varphi_R(y) := \overline{\varphi}_R(y-x)$ and consider the Fourier transform of $p_t^{\eps} \varphi_R$, that is (up to a constant factor)
\[
\hat{p}_{t,R}^{\eps}(\xi) := \widehat{p_t^{\eps} \varphi_R}(\xi) = \esp [ e^{i\langle \xi, X^{\eps}_t \rangle} \varphi_R (X^{\eps}_t) ].
\]
Since the function $y \to p_t^{\eps}(y) \varphi_R(y)$ is $C^{\infty}$ and compactly supported, $\hat{p}_{t,R}^{\eps}$ is integrable and we can use Fourier inversion in order to write 
\be \label{e:FourierInv}
p_t^{\eps}(x) = p_t^{\eps}(x) \varphi_R(x) = \frac1{(2\pi)^n} \int_{\R^n} e^{-i\langle \xi, x \rangle} \hat{p}_{t,R}^{\eps}(\xi)
d\xi.
\ee
On the one hand, it is clear that
\be \label{e:firstFourEstim}
| \hat{p}_{t,R}^{\eps} | \le \Prob(|X^{\eps}_t - x| \le R). 
\ee
On the other hand, using $\partial^k_{x_1} \dots \partial^k_{x_n} e^{i\langle \xi, x \rangle} = i^{kn} (\prod_{j=1}^n \xi^j)^k e^{i\langle \xi, x \rangle}$, and applying Theorem \ref{t:MallWeigEstim}, we have
\be \label{e:FourierEstimate}
\begin{aligned}
\Bigl| \Bigl( 1 + (\prod_{j=1}^n \xi^j)^k \Bigr) \hat{p}_{t,R}^{\eps}(\xi) \Bigr| &=
\Bigl| \Bigl( 1 + (\prod_{j=1}^n \xi^j)^k \Bigr) \esp [ e^{i\langle \xi, X^{\eps}_t \rangle} \varphi_R (X^{\eps}_t) ] \Bigr|
\\
&\le 1 + | \esp[ \partial_{\alpha} e^{i\langle \xi, X^{\eps}_t \rangle} \varphi_R (X^{\eps}_t) ] |
\\
&= 1 + | \esp[ e^{i\langle \xi, X^{\eps}_t \rangle} H_{\alpha}(X^{\eps}_t, \varphi_R (X^{\eps}_t)) ] |
\\
&\le 1 + || H_{\alpha}(X^{\eps}_t, \varphi_R (X^{\eps}_t))||_2
\end{aligned}
\ee
where $\alpha=(\{1\}^k,\dots,\{n\}^k)$.
Using Lemma \ref{l:MallDerEstim} and the fact that the norms $\{ |b_{\eps}|_k \}_{\eps > 0}$ are bounded in $\eps$, it follows that there exist some
$\eps_0 > 0$ such that
\begin{multline} \label{e:indicatorEstim}
||\varphi_R(X^{\eps}_t)||_{nk,2^{n k+1}}
\le
C_k \bigl( 1 + \bigl( t \vee t^k \bigr)^{1/2} \bigr)
\left( t^{1/2} |b_{\eps}|_k + \eps |\sigma|_k \right)^{\gamma'}
e^{k \gamma (t |b_{\eps}|_1 + t^{1/2} \eps |\sigma|_1)^p}
(1 + R^{-nk})
\\
\le 
C^{(1)}_k (1 + R^{-nk})
\end{multline}
for every $\eps < \eps_0$, for some constant $C^{(1)}_k=C^{(1)}_k(t,|b_0|_k, |\sigma|_k)$.

When the SDE \eqref{e:baseSDE} is written in Stratonovich form, the drift $b_{\eps}$ is replaced by $\overline{b}_\eps = b_\eps-\eps^2 \frac12 \sum_{j=1}^d \sum_{k=1}^n \sigma^k_j \partial_{x_k} \sigma_j$.
Noting that $\mathcal{V}_L(x_0^{\eps},v; \overline{b}_\eps, \eps \sigma)$ contains terms propositional to $\eps$ (coming from the brackets $[b_\eps, \eps \sigma_j]_{x_0^{\eps}} = \eps [b_\eps, \sigma_j]_{x_0^{\eps}} = \eps (
[b_0, \sigma_j]_{x_0} + o(1))$) and terms proportional to $\eps^2$ (coming from the brackets $[\eps \sigma_j, \eps \sigma_k]_{x_0^{\eps}}$), for $\eps$ small enough one has
\[
\mathcal{V}_L(x_0^{\eps},v; \overline{b}_\eps, \eps \sigma) \ge 
\frac{\eps}2 \mathcal{V}_L(x_0,v; b_0, \sigma).
\]
Therefore, $\mathcal{V}^{\eps}_L(x_0) :=\inf_{|v|=1} \mathcal{V}_L(x_0^{\eps},v; \overline{b}_\eps, \eps \sigma) \ge \frac{\eps}2 \mathcal{V}_L(x_0; b_0, \sigma)$. 
Under condition (wH), there exist some $L \ge 1$ such that $\mathcal{V}_L(x_0; b_0, \sigma) > 0$.
Then, it follows from estimate \eqref{e:KSMallMatrix} that for every $r > 0$ and $t>0$ there exist $\eps_1$, a function of time $C_r(t)$ and an integer $N(L,n)$ such that
\be \label{e:MallMatrixSmallNoise} 
\esp[ \det (\gamma^{\eps}_t)^{-r} ] \le C_r(t) \: \eps^{-r N(L,n)}
\ee
for every $\eps < \eps_1$, where $\gamma^{\eps}_t$ is the Malliavin covariance matrix of $X^{\eps}_t$.

Now, it follows from \eqref{e:indicatorEstim}, \eqref{e:MallMatrixSmallNoise} and Theorem \ref{t:MallWeigEstim} that
\be \label{e:weightsSmallNoise}
|| H_{\alpha}(X^{\eps}_t, \varphi_R (X^{\eps}_t))||_2 \le
C^{(2)}_k (1 + R^{-nk}) \eps^{-N^1_k}
\ee
for every $\eps < \eps_0 \wedge \eps_1$, for some constant $C^{(2)}_k$ also depending on $(t,|b_0|_k, |\sigma|_k)$, where $N^1_k = r_k N(L,n)$ and $r_k$ is given in Theorem \ref{t:MallWeigEstim}.
Plugging \eqref{e:weightsSmallNoise} into \eqref{e:FourierEstimate}, one obtains a polynomial estimate for $\hat{p}_{t,R}^{\eps}$, precisely
\be \label{e:secondFourEstim}
\hat{p}_{t,R}^{\eps}(\xi) \le \frac{C^{(3)}_k (1+R^{-nk}) \eps^{-N^1_k}}{1+(\prod_{j=1}^n |\xi|^j)^k}.
\ee
Finally, using \eqref{e:weightsSmallNoise} and applying \eqref{e:firstFourEstim} and \eqref{e:secondFourEstim}, for every $q \in (0,1)$ we can write
\[
\begin{aligned}
p_t^{\eps}(x) &=
\frac1{(2\pi)^n} \int_{\R^n} e^{-i\langle \xi, x \rangle} \hat{p}_{t,R}^{\eps}(\xi)
d\xi
\\
&=
\frac1{(2\pi)^n} \int_{\R^n} e^{-i\langle \xi, x \rangle} \hat{p}_{t,R}^{\eps}(\xi)^{q} \
\hat{p}_{t,R}^{\eps}(\xi)^{1-q} d\xi
\\
&\le
\Prob(|X^{\eps}_t - x| \le R)^q \Bigl( C^{(3)}_k (1 + R^{-nk} ) \eps^{-N^1_k} \Bigr)^{1-q} 
\int_{\R^n} \frac1{\bigl( 1+(\prod_{j=1}^n \xi^j)^k\bigr)^{1-q}} d\xi.
\end{aligned}
\]
Choosing $k=k^*$ large enough (but only dependent on $n$ and $q$), the last integral is convergent.
Then, the final claim is proved, once we have set $N_1(q)=\lceil nk^*(1-q) \rceil$ and $N(q)=\lceil N^1_{k^*} (1-q)
\rceil$.
\end{proof}

\bibliographystyle{siam}
\bibliography{References}

\end{document}